\documentclass[12pt]{article}

\usepackage{setspace}
\usepackage[utf8]{inputenc}
\usepackage{amssymb}
\usepackage{amsmath}
\usepackage{mathtools}
\usepackage{amsthm}
\usepackage{thmtools}
\usepackage{geometry}
\usepackage{graphicx}
\usepackage[export]{adjustbox}
\usepackage{wrapfig}
\usepackage{float}
\usepackage{mathrsfs}
\usepackage[colorlinks=true, allcolors=blue]{hyperref}
\usepackage{relsize}
\usepackage[shortlabels]{enumitem}
\usepackage{setspace}
\usepackage{multirow}
\usepackage{subfig}
\usepackage{bbm}
\usepackage[title]{appendix}
\usepackage{tikz}
\usepackage{calc}
\usepackage{bbm}
\usepackage{natbib}

\usetikzlibrary{graphs, graphs.standard}
\usetikzlibrary{decorations.text,decorations.markings}

\usepackage{color} 
\usepackage{listings} 
\usepackage{setspace}
\usepackage{pgfplots}
 
\definecolor{Code}{rgb}{0,0,0} 
\definecolor{Decorators}{rgb}{0.5,0.5,0.5} 
\definecolor{Numbers}{rgb}{0.5,0,0} 
\definecolor{MatchingBrackets}{rgb}{0.25,0.5,0.5} 
\definecolor{Keywords}{rgb}{0,0,1} 
\definecolor{self}{rgb}{0,0,0} 
\definecolor{Strings}{rgb}{0,0.63,0} 
\definecolor{Comments}{rgb}{0,0.63,1} 
\definecolor{Backquotes}{rgb}{0,0,0} 
\definecolor{Classname}{rgb}{0,0,0} 
\definecolor{FunctionName}{rgb}{0,0,0} 
\definecolor{Operators}{rgb}{0,0,0} 
\definecolor{Background}{rgb}{0.98,0.98,0.98} 

\newtheorem{theorem}{Theorem}
\newtheorem{proposition}{Proposition}
\newtheorem{assumption}{Assumption}

\newtheorem{lemma}{Lemma}

\theoremstyle{definition}
\newtheorem{definition}{Definition}
\newtheorem{remark}{Remark}
\newtheorem{examp}{Example}
\newenvironment{example}{\begin{examp}}{\mbox{} \nolinebreak \hfill \mbox{$\Diamond$} \end{examp}}

\renewcommand\thmcontinues[1]{Continued}

\DeclareMathOperator*{\argmax}{argmax}

\usepackage{color,soul}

\begin{document}
\date{November 2025}
\author{Jose M. Betancourt\footnote{\noindent Email: jose.betancourtvalencia@yale.edu. \\
I am grateful to Ben Golub, Tom\'as Rodríguez-Barraquer, Simon Boutin, Sanjeev Goyal, Roberto Corrao, Pau Mil\'an, Sekhar Tatikonda, Elliot Lipnowski, Angelo Mele, Doron Ravid, Aniko \"Ory, Kevin Williams, Annie Chen, Kevin Yin, Laura Tenjo, and Eric Solomon for very informative discussions. I thank seminar participants at the Network Science and Economics Conference for helpful comments and suggestions. Refine.ink was used to proofread the paper for consistency and clarity.
}\\
\small{Yale University}
} 
\title{{\bf The Strength of Local Structures in Decentralized Network Formation}
}

\maketitle

\begin{abstract}
    I study dynamic network formation games in which agents meet stochastically and form links based on their valuation of the network. I show that these games can be represented in terms of the values agents assign to network sub-structures. Particularly, this characterizes potential games as those where all participants in a structure value it equally. When valuations are restricted to a finite set of repeated sub-structures, or \textit{motifs}, the model exhibits \textit{phase transitions}: small changes in motif values cause discontinuous shifts in network density.
    
\end{abstract}

\clearpage

\section{Introduction}
Economic outcomes often depend on who interacts with whom -- whether through trade, information exchange, or collaboration. These structures often set the stage for how agents strategically interact with each other, but they are themselves evolving through time, responding to agents' incentives to form or sever connections. Understanding the forces that drive these dynamics is a central part of describing any economic system with complex interactions. The main challenge with tackling this problem is the vast number of possible structures through which agents can interact. For example, there are around $10^7$ possible networks with $5$ agents, while the number of networks with $20$ agents exceeds the number of atoms in the universe.

The process of forming a network often involves a combination of an element of randomness and strategic decisions. For example, one might randomly meet a friend of a friend at a social gathering, but the decision of whether or not to develop the relationship after this initial meeting is a conscious choice. Such interactions and choices happen constantly, making networks highly dynamic and fluid objects. A model that incorporates random meetings and strategic decisions gives rise to a time-dependent probability distribution over the space of networks. The goal of this paper is to study how the properties of this distribution depend on the incentives of agents to form the network.\footnote{Throughout the paper I focus on directed networks. Most of the results have counterparts in the corresponding undirected case, but require additional structure like transferable utilities.}

To study the properties of dynamic network formation processes, I begin by characterizing some general properties of static network formation games. I show that any network formation game, where agents assign utilities to networks, has an equivalent representation where agents derive value from the sub-structures that are realized in the network. This means that the choice of forming/severing a link depends on the values of the structures that get created/destroyed in the process. This gives a simple economic interpretation to the values of structures: they capture the incentives of agents to deviate by creating or destroying these structures.

The interpretation that structure values capture the strategic incentives of players has an implication that greatly simplifies the analysis of the game. We say that an agent participates in a structure if they have an outgoing link in it. I show that all agents that participate in a structure give it the same value if and only if the game is an exact potential game, meaning that there is a single function that captures the incentives of players to deviate. This is the first technical contribution of the paper. By focusing on potential games, we can characterize the game by specifying a single value to each structure.

Having established that static games have a structure value representation, I study the implications of this for a dynamic network formation game. Agents meet stochastically and decide whether to form or sever links based on deterministic utilities, a choice that is subject to idiosyncratic shocks. The second technical contribution of the paper is to show that this Markov process is reversible if and only if the static game with the deterministic utilities is a potential game. In this case, the stationary probability of the process can be computed explicitly: the probability of observing a network is proportional to the exponential of its potential. This is called a Gibbs measure, and its structure is ideal for studying how the model primitives affect the long-run properties of the network formation process.

To study the properties of large networks, it is useful to have a scheme that allows the size of the network to be scaled up. Motivated by the fact that potential games have a single value assigned to each structure, I study a class of potential games characterized by \textit{motifs}. A motif is a fixed structure, such that every time the structure is realized in the network (allowing re-labeling of the nodes), agents that participate in it get the same value. By specifying a finite set of motifs, we can study the dynamics of the model for an arbitrary number of players. This captures scenarios in which incentives do not depend on the identity of the agents who interact, but only on the structure of interaction. An example of such a scenario is having a fixed cost of forming links and a payoff when this link gets reciprocated.

The next technical contribution of the paper is to show that the model with motif utilities becomes equivalent (in an appropriate graph limit sense) to an Erd\H{o}s--R\'enyi model, where the parameter of the model is determined by an optimization process. The result follows from the literature on large deviations theory applied to large dense graphs \citep{chatterjee_large_2011, chatterjee_estimating_2013}. This optimization problem corresponds to a trade-off between agents' utilities and a quantity called \textit{entropy}, which takes into account the exponentially increasing size of network space. Introducing multiple types of agents generalizes this result, such that the process instead converges to a stochastic block model. In this case, the linking probabilities are determined by an analogous optimization process. 

The nature of the optimization process allows for the emergence of discontinuities in the asymptotic properties of the process, which I refer to as \textit{phase transitions}. This is the main conceptual result in the paper. Since the parameter in the equivalent Erd\H{o}s--R\'enyi model is the maximizer of some objective function, even a continuous change in the values of motifs can lead to a discontinuous change in this parameter. At these transitions, the network discretely changes from a high-density phase to a low-density phase (and vice versa) as parameters are changed.

To illustrate how these phase transitions might arise naturally in an economic setting, I study a simple model of trade, where firms stochastically meet and decide whether to form trade relationships. Firms must incur some distance-dependent cost to state the intent to trade with another firm, and they obtain gains from trade when there is a mutual intent to form a partnership. If the gains from trade are changed continuously (with a tax, for example), a phase transition can be induced, breaking the trading network apart. More complex trading structures lead to diverse manifestations of phase transitions, illustrating the richness of this framework.

Together with the results on the static network formation game, the results on phase transitions shed light on the very complex process of network formation. The low-dimensional behavior of the asymptotics allows for potential applications of this framework to more complex phenomena where networked interactions are important, such as international trade or inter-bank lending.  These results showcase novel phenomena that are relevant for the dynamics of and interventions on complex network structures.

\subsection{Related Literature}
This paper contributes to the broad literature on the analysis of endogenous network formation models. Early work focused on the properties of deterministic network formation games with finite players \citep{jackson_strategic_1996, bala_noncooperative_2000}. The dynamics of network formation have been analyzed in similar models, motivated in their own right and as a selection device for the large number of equilibria that can arise in static models \citep{bala_noncooperative_2000, jackson_evolution_2002, currarini_economic_2009}.\footnote{See \citet{demange_survey_2005, jackson_social_2011} for more complete surveys of network formation models and their properties.} More recent work has analyzed the behavior of network formation games with large numbers of players, which can exhibit properties such as tipping points \citep{golub_strategic_2010, mele_structural_2017, mele_structural_2022, elliott_supply_2022}.

Another branch of the literature on network formation includes actions beyond link formation/deletion, such as \citet{hsieh_structural_2022, sadler_games_2021}. \citet{badev_nash_2021} considers simultaneous changes to the state of multiple network links. This paper differs from these in that more general utility functions for network formation are considered, but only in settings with single link evaluations and with no additional actions. For the case of forward-looking agents, \citet{dutta_farsighted_2005} establishes equilibrium existence results in a more general setting, but lacks a tractable characterization of the resulting probabilities over the space of networks.

A special class of models that is of interest in econometric estimation are Exponential Random Graph Models (ERGMs)\footnote{See \citet{robins_introduction_2007} for an overview of ERGMs}. In these models, the probability associated with a network $g$ is proportional to $\exp(\Phi(g))$, where $\Phi(g)$ is a linear combination of \textit{sufficient statistics} of the network. The advantage of analyzing these models is twofold. First, it has been shown that they arise naturally from strategic network formation models under some regularity assumptions, which would otherwise be intractable \citep{butts_using_2009, mele_structural_2017, christakis_empirical_2020, chandrasekhar_network_2025}. The second advantage of ERGMs is that they provide a natural tool to analyze large graphs, provided that the sufficient statistics scale properly as the number of agents grows \citep{aristoff_phase_2018, mele_structural_2017, mele_structural_2022}. This paper contributes to the literature on ERGMs by characterizing a model of network formation that provides a microfoundation for distributions with arbitrary functions $\Phi(g)$, and hence any ERGM. The model can then be evaluated to see if the primitives satisfy desirable properties, which would yield criteria to evaluate the use of a given ERGM.

Finally, this paper contributes to the literature on graph limits by providing a characterization of the limiting properties of large networks in terms of model primitives. In the case of ERGMs, the model becomes intractable when this number is large but finite\footnote{Previous work on tackling the intractability of these models includes \citet{boucher_my_2017}, which uses Markov Random Fields to reduce the complexity in estimation, and \citet{graham_econometric_2017}, which proposes consistent estimators of structural parameters that take into account degree heterogeneity. \citet{mele_structural_2017} and \citet{mele_structural_2022} use an exchange Markov chain Monte Carlo method to estimate parameters through maximum likelihood in ERGMs.}. However, the distributions can be characterized in the limit of a large number of players. Results from the theory of graph limits \citep{chatterjee_large_2011, chatterjee_estimating_2013, mele_structural_2017} allow for a characterization of the limiting behavior of the process. This characterization makes it evident that the model can exhibit phase transitions\footnote{Phase transitions are a central topic of analysis in statistical physics. See, for example, \citet{holyst_phase_2000, dorogovtsev_critical_2008, squartini_breaking_2015, cimini_statistical_2019} for applications of statistical physics to stochastic network analysis.}, a phenomenon that has been observed in the context of social network formation \citep{golub_strategic_2010} and supply chain formation \citep{elliott_supply_2022}. 
\section{A Static Network Formation Game} \label{sec:static_game}
In this section I present a static model of network formation that will serve as a benchmark for the dynamic game developed in Section \ref{sec:dynamic_game}. I characterize the important properties of the structure of the static game that allow for a tractable study of the dynamic game.

\subsection{Strategies and payoffs}
I consider a game with $N \ge 2$ agents, where ${\cal J}_N \coloneqq \{1,\ldots,N\}$ denotes the set of players. In this game, agents choose who to connect to, and their payoffs are determined by the resulting structure of connections.

An agent's \textit{strategy} is the subset of ${\cal J}_N$ that she chooses to connect to.\footnote{I focus on pure strategies throughout the paper.} I assume that agents cannot form connections with themselves, so agent $i$'s action set is ${\cal S}_{i,N} \coloneqq 2^{{\cal J}_N \backslash \{i\}}$. A typical strategy is denoted with $s_i$. 

A realization of all players' strategies can be characterized as a \textit{network}. I use the notation $ij$ to refer to the tuple $(i,j)$. A (directed) network $g$ is the set of all dyads $ij$ such that $j \in s_i$. Since I do not consider self-interactions, the set of possible dyads is ${\cal D}_N \coloneqq \{ij \in {\cal J}_N^2 \, | \, i \ne j\}$. The set of all networks is, then, ${\cal G}_N \coloneqq 2^{{\cal D}_N}$. For a given network $g$, $g_{-i}$ denotes the subnetwork for which all links with $i$ as the source are removed, such that $g_{-i}$ captures the strategies of all other players. Abusing notation, I often write the network as $g = (s_i, g_{-i})$ or $g = (s_1, \ldots, s_N)$.

Since the strategies of all agents can be represented as a network, agents' payoffs are functions of the realization of the network. Specifically, let $U : {\cal J}_N \times {\cal G}_N \to \mathbb{R}$ be the payoff function, such that $U_i(g)$ is the payoff of agent $i$ under network $g$. I also refer to $U_i(g)$ as the utility that agent $i$ associates to $g$.

\subsection{Potential games}
The network formation game described above is fairly general, admitting $|{\cal G}_N| = 2^{N(N-1)}$ distinct action profiles. Because of its complexity, studying properties of the general game (such as its equilibria) becomes intractable for large $N$. In this section I characterize a special subset of games, called \textit{potential games} \citep{monderer_potential_1996}, which have a structure that allows for a more tractable study of their properties. I also state necessary and sufficient conditions for the network formation game to be a potential game.

\begin{definition} \label{def:potential_game}
    The network formation game $({\cal J}_N, {\cal G}_N, U)$ is a potential game if there exists a function $\Phi: {\cal G}_N \to \mathbb{R}$, called the \textit{potential}, such that
    \begin{align}
        U_i(s_i, g_{-i}) - U_i(s_i',g_{-i}) = \Phi(s_i, g_{-i}) - \Phi(s_i',g_{-i})
    \end{align}
    for all $i \in {\cal J}_N$, $g \in {\cal G}_N$ and $s_i, s_i' \in {\cal S}_{i,N}$.\footnote{More specifically, this is a \textit{cardinal} potential game. In an \textit{ordinal} potential game only the signs of these differences would have to match.}
\end{definition}

Potential games impose additional structure on the network formation game. This allows us to aggregate the incentives for unilateral deviation into a single function, as opposed to study each individual's payoff function. The existence of a potential yields some nice properties of the Nash equilibria of the game.

\begin{remark}
    If the network formation game is a potential game, then all Nash equilibria are local maxima of the potential. That is, a network $g = (s_1, \ldots, s_N)$ is a Nash equilibrium if and only if $\Phi(s_i, g_{-i}) \ge \Phi(s_i', g_{-i})$ for all $i$ and all $s_i'$. Additionally, since ${\cal G}_N$ is finite, there exists a global maximum of $\Phi$, so a Nash equilibrium in pure strategies always exists for potential games.
\end{remark}

Given the added simplicity in studying properties of potential games, it is useful to have conditions under which an arbitrary network formation game is a potential game. This is difficult for general games, but the structure of a network formation game allows for a simplified characterization. In order to provide this characterization, it is useful to define a single-link ``switching'' function. Consider the function $\tau: {\cal D}_N \times {\cal G}_N \to {\cal G}_N$ defined by
\begin{align}
    \tau_{ij}(g) =
    \begin{cases}
        g \cup \{ij\} & \textrm{if } ij \not \in g \\
        g \backslash \{ij\} & \textrm{if } ij \in g.
    \end{cases}
\end{align}
Intuitively, the operation $\tau_{ij}$ creates the link $ij$ if it is not present in a network and severs it if it is. These switching operations are useful to characterize single-link deviations in strategies, rather than an arbitrary change in the composition of an agent's connections. With these single-link deviations, we can define the following condition on the payoff functions.

\begin{definition} \label{def:conservative}
    A set of payoff functions $U$ is said to be \textit{conservative} if the following condition on marginal utilities (MUs) holds:
    \begin{align} \label{eq:conservative_utility}
        &\underbrace{U_i(\tau_{ij}(g)) - U_i(g)}_{\textrm{MU of changing link $ij$ in network $g$}} + \underbrace{U_{i'}(\tau_{i'j'} \circ \tau_{ij}(g)) - U_{i'}(\tau_{ij}(g))}_{\textrm{MU of changing link $i'j'$ in network $\tau_{ij}(g)$}} \nonumber \\
        &= \underbrace{U_{i'}(\tau_{i'j'}(g)) - U_{i'}(g)}_{\textrm{MU of changing link $i'j'$ in network $g$}} + \underbrace{U_{i}(\tau_{ij} \circ \tau_{i'j'}(g)) - U_{i}(\tau_{i'j'}(g))}_{\textrm{MU of changing link $ij$ in network $\tau_{i'j'}(g)$}}
    \end{align}
    for all $g \in {\cal G}_N$ and all $ij, i'j' \in {\cal D}_N$. That is, the sum of the marginal utilities of changing the state of links $ij$ and $i'j'$ is independent of the order in which the changes are made.
\end{definition}

Utility functions satisfying the property in Definition \ref{def:conservative} are called conservative in analogy to conservative vector fields.\footnote{A $C^1$ vector field ${\bf F}:U \to \mathbb{R}^n$, with $U \subset \mathbb{R}^n$ open and simply connected, is said to be conservative if there exists a function $\phi$ such that ${\bf F}({\bf x}) = \nabla \phi({\bf x})$. According to the Poincar\'e Lemma, if $\forall i,j \in \{1,\ldots,n\}$, $\frac{\partial F_i({\bf x})}{\partial x_j} = \frac{\partial F_j({\bf x})}{\partial x_i}$, then ${\bf F}$ is conservative \citep{warner_foundations_1983}.\label{foot:conservative}} The condition in Equation \eqref{eq:conservative_utility} can be thought of as the condition for the Poincar\'e Lemma in Footnote \ref{foot:conservative}, which allows for the characterization of a vector field as the gradient of a scalar field. That is, the information contained in a vector-valued function can be extracted from a scalar-valued function, effectively reducing the dimensionality of the problem. Proposition \ref{prop:conservative} connects this intuition to our network formation game.

\begin{proposition} \label{prop:conservative}
    The network formation game is a potential game if and only if the payoff functions are conservative.
\end{proposition}

The equivalence between conservativeness and the existence of a potential is analogous to Theorem 2.8 in \citet{monderer_potential_1996}, but it is greatly simplified by using the network structure of the game. Since conservativeness is a necessary and sufficient condition for the game to be a potential game, it suffices to check conservativeness to evaluate whether a potential exists. We can go even further, characterizing the entire class of potential games. In order to do this, I use an alternate representation of the incentive structure of the game.

\subsection{The value of group structures}
To begin the characterization of the class of potential games, I use the following technical lemma.

\begin{lemma} \label{lem:value_representation}
    For a network formation game $({\cal J}_N, {\cal G}_N, U)$, there exists a unique function $V:{\cal J}_N \times {\cal G}_N \to \mathbb{R}$ such that
    \begin{align}
        U_i(g) = \sum_{g' \subseteq g} V_i(g')
    \end{align}
    for all $i \in {\cal J}_N$ and $g \in {\cal G}_N$. Additionally, this relation can be inverted as follows:
    \begin{align}
        V_i(g) = \sum_{g' \subseteq g} (-1)^{|g \backslash g'|} U_i(g').
    \end{align}
\end{lemma}

I refer to the quantity $V_i(g)$ as the \textit{value} that $i$ assigns to the \textit{structure} $g$. Despite their similarity to the payoff functions $U$, these objects are conceptually very different. To see this, let us consider the change in utility from adding a single link: $U_i(\tau_{ij}(g)) - U_i(g)$, with $ij \notin g$. This marginal utility only depends on the value $V_i(g')$ of structures $g'$ that contain $ij$ and other links already present in $g$. A similar intuition holds for an arbitrary change in strategy, where the utility of deviating depends only on the value of the structures that are created or destroyed.

The representation of marginal utilities as the sum of values of structures that are created or destroyed gives a very nice representation of conservativeness. Let ${\cal J}_{\textrm{src}}(g) \coloneqq \{i \, | \, \exists j : ij \in g \}$ be the set of \textit{source nodes} of $g$. I also say that agent $i$ \textit{participates} in the structure $g$ if $i \in {\cal J}_{\textrm{src}}(g)$. The following Proposition gives necessary and sufficient conditions for the game to be a potential game in terms of the values of structures.

\begin{proposition} \label{prop:value_representation}
    The network formation game $({\cal J}_N, {\cal G}_N, U)$ is a potential game if and only if
    \begin{align} \label{eq:value_potential_cond}
        V_i(g) = V_j(g) \eqqcolon V_0(g) \quad \forall i,j \in {\cal J}_{\textrm{src}}(g).
    \end{align}
    Furthermore, the potential is given by
    \begin{align}
        \Phi(g) = \sum_{g' \subseteq g} V_0(g').
    \end{align}
\end{proposition}

Proposition \ref{prop:value_representation} states that the network formation game is a potential game if and only if all structures give the same value to all agents who are source nodes in the structure. This means that the path-independence property in Proposition \ref{prop:conservative} takes the form of a symmetry in the values of structures. But, what is special about the source nodes of the structures? Why doesn't this condition apply to all agents? The intuition behind this lies in a property called \textit{choice-equivalence}.

\begin{definition}
    Two sets of utility functions $U$ and $\Tilde{U}$ are \textit{choice-equivalent} if
    \begin{align}
        U_i(s_i,g_{-i}) - U_i(s_i',g_{-i}) = \tilde{U}_i(s_i,g_{-i}) - \tilde{U}_i(s_i',g_{-i})
    \end{align}
    for all $g \in {\cal G}_N$, $i \in {\cal J}_N$ and $s_i,s_i' \in {\cal S}_{i,N}$.
\end{definition}

If two sets of utility functions are choice-equivalent, then agents have the same incentives to change strategies under both of them. This means that they share important properties, such as their Nash equilibria and their potential (if the game is a potential game). Choice-equivalence manifests itself in a clear way when utilities are represented in terms of structure values, as shown in the following Lemma.

\begin{lemma} \label{lem:choice_equivalence}
    Two sets of utility functions $U$ and $\tilde{U}$, with corresponding structure values $V$ and $\tilde{V}$ are choice-equivalent if and only if
    \begin{align} \label{eq:value_choice_equivalence}
        V_i(g) = \tilde{V}_i(g) \quad \forall i \in {\cal J}_{\textrm{src}}(g)
    \end{align}
    for all $g \in {\cal G}_N$.
\end{lemma}

Lemma \ref{lem:choice_equivalence} sheds light on what is special about source nodes. Take some structure $g$ and some $i \notin {\cal J}_{\textrm{src}}(g)$. Then the value $V_i(g)$ can be changed arbitrarily and the resulting utilities will be choice-equivalent to the original ones. This means that for a given agent $i$, the only structures that matter for choice-equivalence (and for the game to be a potential game) are those for which it is a source node.

\begin{example}[label=ex:simple_trade]
    To build an intuition for the results above, I now present a simple model of trade, which I generalize in Section \ref{sec:trade}. Consider $N$ firms which can choose to form trade links with each other. Forming a trade link costs $c > 0$, and if two firms form links with each other, both of them get some value $v > 0$. Explicitly, firm $i$'s utility function is
    \begin{align}
        U_i(g) = v \underbrace{\sum_{j \in {\cal J}_N} \mathbbm{1}\{ij \in g, ji \in g\}}_{\textrm{reciprocated links}} - c \underbrace{\sum_{j \in {\cal J}_N} \mathbbm{1}\{ij \in g\}}_{\textrm{outgoing links}}.
    \end{align}
    In this example, there are two types of structures to which agents assign value: outgoing links and reciprocated pairs of links. Since both firms get the same value from trade, the conditions in Proposition \ref{prop:value_representation} hold, and the potential is
    \begin{align} \label{eq:simple_trade_potential}
        \Phi(g) = v \underbrace{\frac{1}{2} \sum_{i,j \in {\cal J}_N} \mathbbm{1}\{ij \in g, ji \in g\}}_{\textrm{total reciprocated pairs}} - c \underbrace{\sum_{i,j \in {\cal J}_N} \mathbbm{1}\{ij \in g\}}_{\textrm{total outgoing links}}
    \end{align}
    The factor of $1/2$ ensures that we are counting the number of appearances of a reciprocated pair, instead of all the links that are being reciprocated.
\end{example}

With this analysis, we can fully specify all utility functions for which the game is a potential game. 

\begin{remark}
    Suppose the network formation game $({\cal J}_N, {\cal G}_N, U)$ is a potential game with potential $\Phi$. Define the structure values
    \begin{align}
        V_0(g) \coloneqq \sum_{g' \subseteq g} (-1)^{|g \backslash g'|}\Phi(g').
    \end{align}
    Then $U$ is choice-equivalent to the utility functions
    \begin{align}
        \tilde{U}_i(g) \coloneqq \sum_{g' \subseteq g} V_0(g') \mathbbm{1}\{i \in {\cal J}_{\textrm{src}}(g')\}.
    \end{align}
    Therefore, all utilities corresponding to potential games can be characterized, up to choice-equivalence, in terms of the potential.
\end{remark}

Having characterized the structure of potential games in the static network formation game, I now formulate a model of dynamic network formation. The structure of the static game will be useful to characterize the properties of the dynamic game.

\section{Dynamic Network Formation} \label{sec:dynamic_game}
In this section I study a dynamic process of network formation. In this model, agents meet stochastically and make choices on whether to change the state of their connections. Throughout this section, the set of agents is ${\cal J}_N$ and the space of networks is ${\cal G}_N$

\subsection{Stochastic meeting of agents}
The network formation process occurs in continuous time. The state of the network at time $t$ is $G_t \in {\cal G}_N$. Agents meet stochastically and, conditional on meeting, decide whether they want to change the state of the network. If the network state is $g$, agent $i$ meets agent $j$ at an exogenous Poisson rate $\lambda_{ij}(g)$, that can potentially depend on $g$. I impose the following restriction on the meeting rates:
\begin{assumption} \label{A:meeting_rates}
    The meeting rates satisfy $\lambda_{ij}(g) = \lambda_{ij}(\tau_{ij}(g))$ for all dyads $ij$ and all networks $g$.
\end{assumption}
Assumption \ref{A:meeting_rates} means that the meeting rate for the dyad $ij$ is independent of whether it is present in the network, but can depend on the rest of the network. This is useful to obtain a tractable characterization of the stochastic dynamics.

Conditional on meeting $j$, agent $i$ makes a decision to keep or change the state of their relationship through a discrete choice rule, parametrized by a set of utility functions $U: {\cal J}_N \times {\cal G}_N \to \mathbb{R}$. Agent $i$ has a baseline utility $U_i(g)$ associated with network $g$, and there are idiosyncratic shocks to this utility every time a choice is made. After meeting at time $t$, agent $i$ chooses to change the state of link $ij$ if and only if
\begin{align} \label{eq:disc_choice}
    (1-\sigma) U_i(\tau_{ij}(G_t)) + \sigma \varepsilon_{t}^0 \ge (1-\sigma) U_i(G_t) + \sigma \varepsilon_{t}^1.
\end{align}
where the shocks $\varepsilon_{t}^k$ are i.i.d. (across $k$ and across time) drawn from some distribution $F_0$ and $\sigma \in (0,1)$ modulates the importance of noise in the choice process. As $\sigma \to 0$, the process becomes deterministic and the highest payoff option is chosen. As $\sigma \to 1$, the process becomes fully random and the state of each link is changed with probability $1/2$. Let $F_1$ be the CDF of $\varepsilon^1_{t} - \varepsilon^0_{t}$, such that the probability of the network changing from $g$ to $\tau_{ij}(g)$, denoted by $p_{ij}(g)$, is given by
\begin{align}
    p_{ij}(g) = 1 - F_1\left[ \left( \frac{1-\sigma}{\sigma} \right) (U_i(\tau_{ij}(g)) - U_i(g)) \right].
\end{align}
With this characterization, we can now analyze how the dynamics of the model depend on the structure of the utility functions.

\subsection{Dynamics and convergence}
I assume at time $t=0$ there is some initial probability distribution $\pi_0 \in \Delta({\cal G}_N)$ over the space of networks. Let $\pi_t(g)$ be the probability of having $G_t = g$. The measure $\pi_t$ evolves according to the Kolmogorov forward equation:
\begin{align}
    \dot{\pi}_t(g) = \sum_{ij \in {\cal D}_N} \lambda_{ij}(g) \left[p_{ij}(\tau_{ij}(g)) \pi_t(\tau_{ij}(g)) - p_{ij}(g) \pi_t(g) \right].
\end{align}
Note that, as long as the distribution of utility shocks $F_0$ has full support, the switching probabilities $p_{ij}(g)$ are non-degenerate. Therefore, in this case, the network formation process corresponds to an aperiodic irreducible Markov process on ${\cal G}_N$. It follows from standard results in stochastic processes (e.g. Theorem 1.19 in \citet{durrett_essentials_2016}) that there is a unique distribution $\pi$ such that $\pi_t(g) \to \pi(g)$ for all $g \in {\cal G}_N$, regardless of the initial conditions. I now include additional structure in the stochasticity of choices, which allows tractability of the analysis to follow.
\begin{assumption} \label{A:T1EV}
    The shocks $\varepsilon_{t}^k$ follow a type 1 extreme value (T1EV) distribution.
\end{assumption}

Assumption \ref{A:T1EV} ensures that shocks have full support, which guarantees convergence of the probability distribution. Additionally, it implies that the switching probabilities $p_{ij}(g)$ are a logistic function of the marginal utility of switching $U_i(\tau_{ij}(g)) - U_i(g)$. Equivalently, they satisfy
\begin{align} \label{eq:log_odds_ratio}
    \log \left( \frac{p_{ij}(g)}{p_{ij}(\tau_{ij}(g))} \right) = \left( \frac{1-\sigma}{\sigma} \right) (U_i(\tau_{ij}(g)) - U_i(g)).
\end{align}
We see that there is an intuitive connection between the properties of the Markov chain (through the transition probabilities) and the utility functions of the agents. Given this connection, if the static game with utility functions $U$ has some nice structure, we would expect the dynamic game to inherit some of this structure. This is formalized in the following Proposition for the case of potential games.

\begin{proposition} \label{prop:conv}
    Consider the static game $({\cal J}_N, {\cal G}_N, U)$. This game is a potential game if and only if the Markov chain of the dynamic process is reversible. Furthermore, if the potential of the static game is $\Phi$, then the stationary distribution of the process is given by the Gibbs measure
    \begin{align} \label{eq:eq_dist}
        \pi(g) = \frac{\exp\left[ \left( \frac{1-\sigma}{\sigma} \right) \Phi(g)\right]}{\sum_{g' \in {\cal G}_N} \exp\left[ \left( \frac{1-\sigma}{\sigma} \right)\Phi(g')\right]}.
    \end{align}
\end{proposition}

Proposition \ref{prop:conv} states that the same structure that allows for a tractable study of incentives in potential games allows us to fully characterize the stationary distribution in the dynamic game. This is because reversibility is mathematically equivalent to conservativeness.\footnote{Reversibility of a Markov jump process on a finite space ${\cal X}$ is usually stated in terms of the detailed balance condition: $W(i,j) \pi(i) = W(j,i) \pi(j)$ for all $i,j \in {\cal X}$, where $W(\cdot)$ are the transition rates and $\pi(\cdot)$ are the stationary probabilities. An equivalent characterization solely in terms of rates is Kolmogorov's criterion, which is analogous to conservativeness \citep{norris_continuous-time_1997}.} Reversibility allows us to explicitly construct the stationary distribution in terms of a function that encodes the transition rates of the stochastic process.\footnote{In the theory of Markov processes, this function is also called a potential, or energy, function. For example, similar dynamics can be used to sample the Boltzmann distribution of the Ising model \citep{glauber_time-dependent_1963}.} Interestingly, it is precisely the potential function of the static game that determines the stationary distribution of the network formation process. Throughout the rest of the paper I exploit this form of the stationary distribution to study the properties of the dynamic game.

This characterization of the stationary distribution is related to the family of exponential random graph models (ERGMs). These assign a probability to a network of the form
\begin{align} \label{eq:ERGM}
    \pi_{\textrm{ERGM}}(g) = \frac{\exp(\boldsymbol{\beta} \cdot {\bf S}(g))}{\sum_{g' \in {\cal G}} \exp(\boldsymbol{\beta} \cdot {\bf S}(g'))},
\end{align}
where ${\bf S}(g)$ is a vector of sufficient statistics of the network and $\boldsymbol{\beta}$ is a vector of model parameters. Therefore, we see that if the potential of our network formation game is of the form $\Phi(g) = \boldsymbol{\beta} \cdot {\bf S}(g)$, the stationary distribution will be the distribution of an ERGM with sufficient statistics ${\bf S}$. This means that this myopic network formation model serves as a microfoundation for ERGMs, conditional on finding suitable utility functions\footnote{\citet{mele_structural_2017} provides a microfoundation for a class of ERGMs with this model. \citet{chandrasekhar_network_2025} generalize this analysis to a wider class of utility functions.}. One property of ERGMs that is shared by the myopic network formation game is that the denominator in the Gibbs measure involves a sum over all possible networks, making its calculation intractable. This intractability causes problems for estimation of ERGMs (see \citet{chandrasekhar_network_2025} for a detailed description of the problems with estimating ERGMs and potential solutions). However, as I discuss in Section \ref{sec:large_networks}, for a certain class of models where the number of agents can be systematically increased, characterizing the large $N$ behavior of this denominator allows for even further dimensionality reduction of the model.

\subsection{Equilibrium selection}
It is interesting to see how properties of the static game manifest in the dynamics of network formation. For the case of choice-equivalent utilities, we have the following: 

\begin{remark}
    If two sets of utility functions $U$ and $\Tilde{U}$ are choice-equivalent, then the switching probabilities $p_{ij}(g)$ are the same under $U$ and $\Tilde{U}$ for all dyads $ij$ and networks $g$.
\end{remark}

This means that the dynamics under two sets of choice-equivalent utility functions are \textit{identical}. This helps to strengthen the intuition for why they must have the same potential, since they must have the same stationary distribution.

We can also analyze what networks get ``selected'' by the dynamic process. As a benchmark, we can analyze the case of $\sigma = 0$, for which the dynamics have no idiosyncratic shocks. As shown in Proposition 2 of \citet{mele_structural_2017}, Nash equilibria are absorbing states of these dynamics, and the network converges to one of the Nash equilibria with probability 1. Therefore, the process with no shocks selects the Nash equilibria of the game.

Our characterization of the dynamics allows us to go beyond the no-shocks case for potential games. Previous work has used dynamic models to select equilibria which are robust to specific kinds of perturbations.\footnote{See, for example, \citet{bala_noncooperative_2000} and \citet{jackson_evolution_2002}.} In our case, we can check which equilibria are present for small but non-zero $\sigma$.

\begin{remark}
    As $\sigma \to 0^+$, the stationary distribution of the process converges to the uniform distribution across the potential-maximizing Nash equilibria (PMNE) of the game. A consequence of this is that if $g \notin \argmax_{g'} \Phi(g')$, then $\pi(g) \to 0$.
\end{remark}

This behavior of the stationary distribution \textit{refines} the set of Nash equilibria of the game. This is because the set of PMNE of the game are robust (in the sense of having a non-vanishing stationary probability) to small shocks in the choice process. Additionally, we can get some intuition for the properties of the dynamic process by analyzing the PMNE of the game.

\begin{example}[continues=ex:simple_trade]
    Let us analyze the properties of the Nash equilibria and PMNE in our simple static trade example. For a given set of incoming links, an agent will want to either reciprocate all of them or not form any outgoing links, based on the value of $v$ relative to $c$. This means that the Nash equilibria of the game are:
    \begin{itemize}
        \item the empty network if $v < c$,
        \item all networks $g$ where $ij \in g \iff ji \in g$ if $v \ge c$.
    \end{itemize}
    To obtain the set of PMNE, we can use the potential computed in Equation \eqref{eq:simple_trade_potential}. The PMNE, then, are
    \begin{itemize}
        \item the empty network if $v < 2 c$,
        \item the complete network if $v > 2 c$,
        \item all networks $g$ where $ij \in g \iff ji \in g$ if $v = 2 c$.
    \end{itemize}
    We can see that the set of PMNE drastically changes as $v$ changes from being less than $2c$ to being larger. This intuition behind the drastic change in the PMNE is related to the drastic changes we observe for large networks in the next section.
\end{example}

\section{Large Dense Networks} \label{sec:large_networks}
The results in the previous sections gave conditions under which our network formation process can be characterized in terms of a potential, and established some properties of the structure of such processes. This allowed us to identify the most probable networks with PMNE of the deterministic game. In this section, I study processes generated by a particular class of potential games as $N$ grows large, revealing interesting typical behaviors that arise from the interplay of incentives and the exponentially increasing size of network space.

\subsection{Network motifs}
As illustrated by Proposition \ref{prop:value_representation}, conservativeness of utilities is tightly related to the value of social structures. We can study the large-$N$ behavior of the model for games where there are a few structure \textit{types} from which people derive utility. For example, in Example \ref{ex:simple_trade}, firms derive value from outgoing links or reciprocated pairs. We can generalize this by fixing a few structure types and their values, and having agents derive this value every time they participate in one of these structures. Following the graph theory literature, I call these structures \textit{motifs}. This characterization is particularly amenable to our analysis, since it creates a clean connection between the large $N$ asymptotics of the networks and the model primitives.

Let us begin by fixing the number of agents $N$. Formally, a motif $m$ is a network over a set ${\cal J}_m = \{1, \ldots, n_m\}$ of $n_m$ nodes that has $e_m \coloneqq |m|$ edges. We say that there is an instance of motif $m$ in network $g$ if there exists an injective map $\varphi: {\cal J}_m \to {\cal J}_N$ such that $\varphi(m) \subseteq g$, where $\varphi(m) \coloneqq \{\varphi(i) \varphi(j): ij \in m\}$ is the subnetwork induced by $\varphi$.

I will focus on dense networks, where the density of the network remains bounded away from $0$ almost surely as $N \to \infty$.\footnote{This is because of the large body of work on the asymptotics of dense networks \citep{lovasz_large_2012}. It is definitely of interest to study how similar models behave for sparse networks.} In order to obtain the correct scaling, I assume that agents derive value $a_m/N^{n_m-2}$ from each realization of motif $m$ that they are a source node of. If the value scaled faster, we would get a degenerate behavior in the asymptotics: either the network density would converge to $0$ if $a_m < 0$ or to $1$ if $a_m > 0$. If the value scaled slower, then the motif would not be relevant enough in influencing the behavior of agents, and the network density would converge to $1/2$, as if the linking choices were uniformly random.

To obtain the total value that an agent $i$ derives from a motif $m$, we need to calculate how many instances of the motif she participates in. For this, it is useful to count these instances and relate this count to the case of the complete network. There are $N^{n_m}$ maps $\varphi: {\cal J}_m \to {\cal J}_N$, and the fraction of these that are injective converges to $1$ as $N \to \infty$. Now consider the complete network. Every injective map results in a realization of $m$ so, due to symmetry, agent $i$ will participate in a fraction $|{\cal J}_{\textrm{src}}(m)|/N$ of the realizations of $m$ in the complete network. With this, we can define a participation density of $i$ in $m$.

\begin{definition}
    The \textit{participation density} of individual $i$ in motif $m$ given network $g$ is
    \begin{align}
        b_{i,N}(m,g) \coloneqq \frac{N}{|{\cal J}_{\textrm{src}}(m)|} \frac{1}{N^{n_m}} \sum_{\substack{\varphi: {\cal J}_m \to {\cal J}_N \\ \varphi \textrm{ injective}}} \underbrace{ \mathbbm{1}\{\varphi(m) \subseteq g\}}_{\substack{\textrm{motif is present} \\ \textrm{in network}}} \underbrace{\mathbbm{1}\{i \in {\cal J}_{\textrm{src}}(\varphi(m))\}}_{\substack{i \textrm{ is a source node of} \\ \textrm{ the motif}}}.
    \end{align}
    Similarly, the \textit{subgraph density} of motif $m$ in network $g$ is
    \begin{align}
        b_N(m,g) \coloneqq \frac{1}{N^{n_m}} \sum_{\substack{\varphi: {\cal J}_m \to {\cal J}_N \\ \varphi \textrm{ injective}}} \underbrace{ \mathbbm{1}\{\varphi(m) \subseteq g\}}_{\substack{\textrm{motif is present} \\ \textrm{in network}}} = \frac{1}{N} \sum_{i \in {\cal J}_N} b_{i,N}(m,g).
    \end{align}
\end{definition}

These definitions are useful since $b_{i,N}, b_N \in [0,1)$ and $b_{i,N}(m,g) = b_N(m,g) = 0$ for the empty network and $b_{i,N}(m,g), b_N(m,g) \to 1$ as $N \to \infty$ for the complete network. This allows us to characterize the network in terms of these participation densities, which will directly enter the utilities. 

\begin{figure}
    \begin{center}
    \caption{Motif counting}
    \includegraphics[width = 0.7 \linewidth]{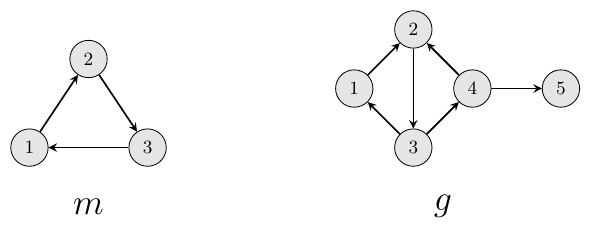}
    \label{fig:motif_counting}
    \end{center}
    \footnotesize{\textit{Notes.} This figure illustrates the process of counting the realizations of motif $m$ in network $g$. There are $3$ injective mappings $\varphi$ from the nodes of $m$ to $\{1,2,3\}$ and $3$ others to $\{2,3,4\}$ such that $\varphi(m) \subseteq g$. The degeneracy of the motif is $\delta_m = 3$ and it has $|{\cal J}_{\textrm{src}}(m)| = 3$ source nodes. Therefore, the participation density for nodes $1$ and $4$ is $b_{i,5}(m,g) = 3/3 \cdot 5^2 = 0.04$, for nodes $2$ and $3$ it is $b_{i,5}(m,g) = 6/3 \cdot 5^2 = 0.08$, and for node $5$ it is $b_{5,5}(m,g) = 0$. The subgraph density of $m$ in $g$ is $b_5(m,g) = 0.048$.}
\end{figure}

To finish computing the number of realizations of $m$ that $i$ participates in, we need to account for the overcounting in how we compute the number of realizations. There can be multiple maps that give the same realization of the motif, as illustrated in Figure \ref{fig:motif_counting}. Let $\delta_m$, the \textit{degeneracy} of motif $m$, be the number of injective maps $\varphi':{\cal J}_m \to {\cal J}_m$ such that $\varphi'(m) = m$. Having corrected for the counting, we have that the number of realizations of $m$ that $i$ participates in is $N^{n_m-1} |{\cal J}_{\textrm{src}}(m)| b_{i,N}(m,g)/\delta_m$. Therefore, the utility that agent $i$ derives from motif $m$ is
\begin{align}
    U_{i,N}^m(g;a_m) = N |{\cal J}_{\textrm{src}}(m)| \frac{a_m}{\delta_m} b_{i,N}(m,g).
\end{align}
If we consider a (finite) set of motifs ${\cal M}$, with values ${\bf a} \coloneqq (a_m)_{m \in {\cal M}}$, the total utility she derives from these motifs is
\begin{align} \label{eq:motif_utility}
    U_{i,N}^{\cal M}(g;{\bf a}) = N \sum_{m \in {\cal M}} |{\cal J}_{\textrm{src}}(m)| \frac{a_m}{\delta_m} b_{i,N}(m,g).
\end{align}
To obtain the potential, we need the number of realizations of the motif in the network, which we can compute analogously with the subgraph density. Therefore, the potential for the game with $N$ players is
\begin{align}
    \Phi_N^{{\cal M}}(g;{\bf a}) = N^2 \sum_{m \in {\cal M}} \frac{a_m}{\delta_m} b_N(m,g).
\end{align}
Having characterized a scheme under which we can study the model with a large number of players, we can now analyze the properties of the resulting networks.

\subsection{Phase transitions} \label{subsec:phase_transitions}
An astonishing result in the theory of graph limits is the possibility of discontinuous changes in the aggregate properties of the model caused by continuous changes in parameters. These types of phenomena fall under the category of \textit{phase transitions} in the physics and applied mathematics literature, and the properties of systems near critical points (the parameter values where phase transitions occur) have been studied extensively\footnote{For a detailed overview of phase transitions in physics, see \citet{goldenfeld_lectures_2018}. For an application of the theory of phase transitions to ERGMs, see \citet{aristoff_phase_2018}.}. In the case of our model, we will see that varying the values of motifs can induce a phase transition, causing a discontinuous change in the typical density of the networks.

The largest difficulty in analyzing the behavior of this model is characterizing the denominator in Equation \eqref{eq:eq_dist}, which is intractable. This is because the number of terms in this sum is $2^{N(N-1)}$. Despite the difficulty in calculating this sum, this object encodes useful information about the system statistics.\footnote{In the physics literature, this denominator is called the \textit{partition function} \citep{huang_statistical_1987}.} I now define the \textit{scaled normalization constant}.

\begin{definition}
    The scaled normalization constant of the system for motifs ${\cal M}$ with values ${\bf a}$ is
    \begin{align}
        \zeta_N^{\cal M}({\bf a},\sigma) \coloneqq \frac{1}{N^2} \log\left( \sum_{g \in {\cal G}_N} \exp\left[ \left( \frac{1-\sigma}{\sigma} \right) \Phi_N^{\cal M}(g ; {\bf a})) \right] \right).
    \end{align}
    The limiting scaled normalization constant is
    \begin{align}
        \zeta^{\cal M}({\bf a}, \sigma) \coloneqq \lim_{N \to \infty} \zeta_N^{\cal M}({\bf a}, \sigma).
    \end{align}
\end{definition}

To understand the $1/N^2$ re-scaling intuitively, note that to obtain dense networks, the participation densities must be typically non-vanishing. This means that the potential must be of order $N^2$. Additionally, note that the size of network space scales as $2^{N^2}$, so we expect the logarithm to scale as $N^2$. Together, these arguments make us expect $\zeta^{\cal M}_N({\bf a}, \sigma)$ to be of order $1$. 

We can now state the main result for the analysis of the motif model. An interesting aspect of the model is that the scaled partition function of the system is not only useful for calculating statistics of the limiting distribution, but rather is the central object that characterizes the typical behavior of the system as it navigates the space of networks.

\begin{proposition} \label{prop:motif_partition}
Suppose $a_m > 0$ for all motifs $m$ with $e_m > 1$. Define the entropy-adjusted potential (EAP) as
\begin{align}
    \Gamma^{\cal M}(\rho ; {\bf a}, \sigma) \coloneqq \left( \frac{1-\sigma}{\sigma} \right) \underbrace{\sum_{m \in {\cal M}} \frac{a_m}{\delta_m} \rho^{e_m} }_{\textrm{motif values}} + \underbrace{H(\rho)}_{\textrm{entropy}}
\end{align}
where $H(p) \coloneqq -p \log(p) - (1-p) \log(1-p)$ is the entropy of a Bernoulli random variable with parameter $p$. Then the limiting normalization constant exists and is given by
\begin{align} \label{eq:motif_partition}
    \zeta^{\cal M}({\bf a}, \sigma) = \max_{\rho \in [0,1]} \Gamma^{\cal M}(\rho ; {\bf a},\sigma).
\end{align}
Additionally, if the maximizer $\rho^*_{{\cal M}}({\bf a}, \sigma)$ is unique, then as $N \to \infty$ the networks generated by the model become indistinguishable from those generated by an Erd\H{o}s--R\'enyi random graph model with parameter $\rho^*_{{\cal M}}({\bf a}, \sigma)$.\footnote{See Appendix \ref{sec:app_graph_limits} for a detailed treatment of graph limits}
\end{proposition}

Proposition \ref{prop:motif_partition} is similar to Theorem 2 in \citet{mele_structural_2017}, and builds on it by extending the result to a microfounded model with arbitrary structures using the characterization in Proposition \ref{prop:value_representation}. This result essentially reduces the dimensionality of the complex network formation process to a single one-dimensional optimization process, where the maximizer $\rho^*$ can be interpreted as the typical density of the resulting networks. The nature of this optimization process is interesting, since no single agent is solving the problem in Equation \eqref{eq:motif_partition}. To understand it, it is useful to analyze where the entropy term comes from. This term arises from the fact that the number of networks over $N$ nodes with density $\rho$, to leading exponential order, is approximately $e^{N^2 H(\rho)}$. Therefore, as agents stochastically try to optimize their utility by building motifs, they get ``stuck'' exploring an exponentially growing region of network space, leading to a trade-off between utility and entropy.

It is important to note that this analysis is restricted to positive motif values for motifs that involve more than one link, since the characterization in terms of typical densities can break down when these values are negative, as shown in \citet{mele_structural_2017}. This is of particular interest when analyzing the problem of estimation of model parameters, but I will restrict my attention to positive values, since this is still a rich enough case where there are non-trivial economic forces at play. One particularly interesting phenomenon that arises in the large $N$ limit is the emergence of phase transitions. We can illustrate this with our simple model of trade.

\begin{example}[continues=ex:simple_trade]
    We can easily see how our simple trade model can be formulated in terms of motifs. In this case, our set of motifs is ${\cal M} = \{m_1, m_2\}$, where $m_1 \coloneqq \{12\}$ is the ``outgoing links'' motif and $m_2 \coloneqq \{12, 21\}$ is the ``reciprocated links'' motif. The values assigned to these motifs are ${\bf a} = (-c,v)$ and their degeneracies are $\delta_{m_1} = 1$ and $\delta_{m_2} = 2$. Note that both motifs have $n_m = 2$, so their value is constant as $N$ grows.

    The entropy-adjusted potential $\Gamma$ is shown in panel (a) of Figure \ref{fig:density_phase_transition} for a fixed value of $c$. Note that $\Gamma$ is a continuous function of all its parameters. For low and high values of $v$, there is a unique local maximum of $\Gamma$. However, for $v$ near $2c$, $\Gamma$ has two local maxima, and their values coincide at $v=2c$. This means that there is a discontinuous jump in the density that maximizes $\Gamma$ at $v=2c$. This happens despite $\Gamma$ being continuous because, while the location of local maxima does change continuously, there can be discrete jumps between two different local maxima.

    As shown in panel (b) of Figure \ref{fig:density_phase_transition}, this discontinuous jump occurs for a wide range of values of $v$, and it defines a clear change between a phase of low-density and high-density networks. Thus, reducing the value of mutual links across the line of phase transitions can lead to the discontinuous collapse of the networks in the model. 
\end{example}

The discrete jumps in the behavior of the typical density in the example have a striking resemblance to the change in the structure of the PMNE of the static game. If we compare the average density of the PMNE networks, we have a density of 0 for the empty network and 1 for the full network. This means that the average density of the PMNE networks jumps from 0 to 1 as $v$ changes from being lower than $2c$ to being higher. This jump seems to remain for non-zero $\sigma$. With this intuition, we can define what a phase transition is, and why there is a connection to the behavior of the set of PMNE.

\begin{figure}
    \begin{center}
    \caption{Phase transition induced by motif values}%
    \subfloat[Entropy-adjusted potential]{{\includegraphics[width=0.49\linewidth]{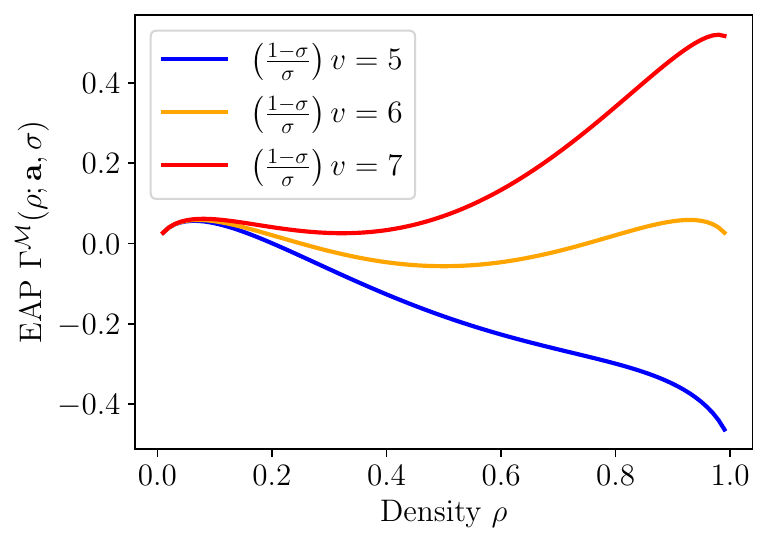} }}%
    \subfloat[Density phase diagram]{{\includegraphics[width=0.49\linewidth]{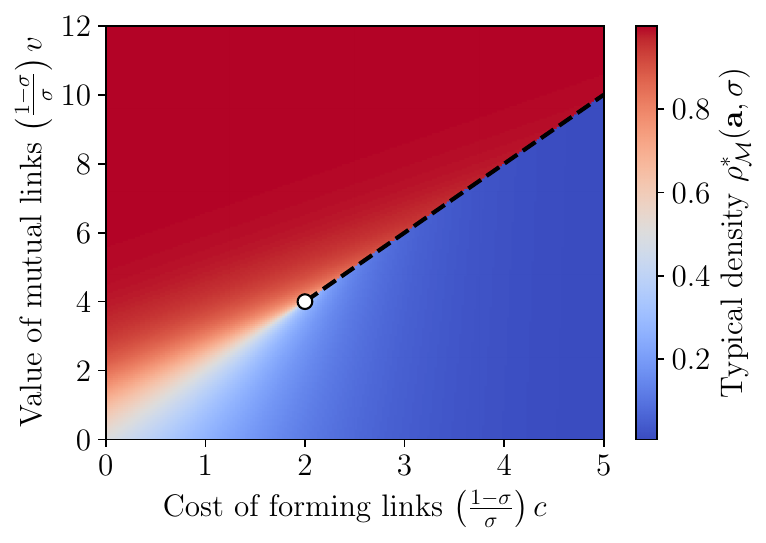} }}
    \label{fig:density_phase_transition}
    \end{center}
    \footnotesize{\textit{Notes.} This figure illustrates the mechanisms behind phase transitions with two motif values. Panel (a) shows the behavior of the entropy-adjusted potential $\Gamma^{\cal M}(\rho;{\bf a},\sigma)$ as the value of mutual links $v$ is varied, while the value of single connections is kept constant at $\left( \frac{1-\sigma}{\sigma} \right) c = 3$. Panel (b) shows the value of the typical density for different values of cost $c$ and the value of mutual-links $v$.}
\end{figure}

\begin{definition}
    The system is said to undergo a \textit{phase transition} at parameters $({\bf a},\sigma)$ if $\rho^*_{{\cal M}}(\bf a, \sigma)$ is discontinuous in ${\bf a}$.\footnote{In the physics literature, there are other kinds of phase transitions that involve non-analyticities, not only discontinuities.} A phase transition at parameters $({\bf a},\sigma)$ is said to be \textit{incentive-driven} if $\lim_{\sigma \to 0^+} \rho^*_{{\cal M}}(\bf a, \sigma)$ is discontinuous in ${\bf a}$.\footnote{This is related to the concept of zero-temperature phase transitions in statistical physics. See \citet{goldenfeld_lectures_2018} for a detailed discussion.} If a phase transition is not incentive-driven, it is said to be \textit{entropy-driven}.
\end{definition}

The intuition behind phase transitions being incentive-driven is that there is a stark change in the structure of the PMNE that survives having non-zero noise. If noise becomes too strong, however, the phase transition can disappear, since incentives are not strong enough to drive the transition. This can be seen in the lower-left corner of Figure \ref{fig:density_phase_transition}, panel (b). We can also find examples of entropy-driven phase transitions, where phase transitions appear when the motifs are sufficiently complex.

\begin{example}
    To illustrate how complexity in the motifs can drive phase transitions, consider a very simple supply chain formation game. To produce a product, $\ell$ firms must come together to form a supply chain. We say a supply chain is formed if a firm connects to another, which connects to another, etc. for a total of $\ell-1$ connections. There are no costs to forming links, just a positive payoff from forming the supply chain (to be specified below). Therefore, since there are no incentives to sever links, the unique Nash equilibrium (and hence PMNE) of the deterministic game is the complete network. This means that this model cannot have an incentive-driven phase transition.

    This game corresponds to a motif model with a single motif: an $\ell$-node directed chain. Formally, the $\ell$-node chain is the network $s_\ell = \{12,23,\ldots,(\ell-1)\ell\}$. The degeneracy of this motif is $\delta_{s_\ell} = 1$. If there are $N$ players, the payoff that each firm gets for each chain it participates in is $v/N^{\ell-2}$ for some $v > 0$. Hence, $v$ is the value of the supply chain motif.

    The typical density of the network as a function of supply chain value is shown in Figure \ref{fig:chain_phase_transition}. The 5-node model shows a continuous change from the $\rho = 1/2$ density obtained from shock-dominated choices to the $\rho= 1$ density from the PMNE. In contrast, the 7-node and 9-node models show a large discontinuity in the typical density as the value of the chain increases. This can be interpreted as a transition from a shock-dominated phase to a utility-dominated phase.
\end{example}

\begin{figure}
    \begin{center}
    \caption{Phase transition in the $\ell$-node chain model}
    \subfloat[The $\ell$-node chain $s_\ell$]{{\includegraphics[width=0.4\linewidth]{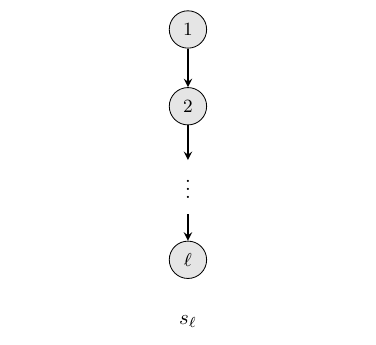} }}%
    \qquad
    \subfloat[Phase transition for various $\ell$-node chains]{{\includegraphics[width=0.49\linewidth]{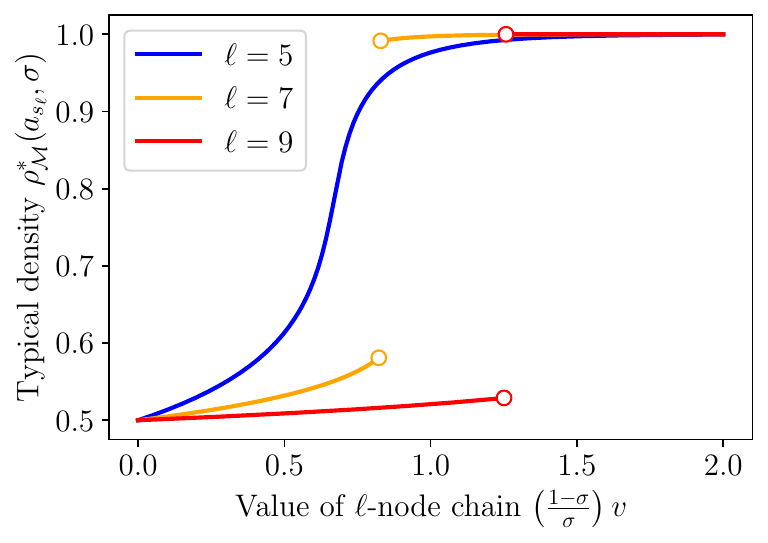} }}
    \label{fig:chain_phase_transition}
    \end{center}
    \footnotesize{\textit{Notes.} This figure shows the behavior of the typical density under the $\ell$-node chain model. Panel (a) shows an illustration of the $\ell$-node chain motif. Panel (b) shows how the typical density depends on the value of the chain $v$ for different chain lengths $\ell$.}
\end{figure}

An important ingredient for phase transitions to appear is that agents must derive utility from structures that are ``sufficiently complex'' to generate these effects. For example, if the only motif we considered were a single link, the typical density would be a continuous function of the utility of this motif. A similar complexity requirement was found in the analysis of supply networks in \citet{elliott_supply_2022}, where there is a discontinuity in the probability of successful production as a function of how likely individual production ties are successful. They find that the complexity of production, defined as the number of different inputs required to produce some good, is fundamental in generating this discontinuity. It seems, therefore, that the feedback generated by having multiple relationships is a recurring feature of models that generate these types of phenomena.

\subsection{Heterogeneous agents} \label{subsec:heterogeneous_agents}
Our previous analysis showed that the network formation model can generate non-trivial phenomena even when all agents have the same utility. In this section, I expand on this by analyzing the network formation game for heterogeneous agents. Heterogeneity is a natural property of real network formation processes. For example, individuals have the tendency to form disproportionately more connections with people similar to themselves than to others. This is a phenomenon known as homophily \citep{mcpherson_birds_2001, bramoulle_stochastic_2016}. Another setting where heterogeneity is important is when networks are spatially embedded, such as in trade routes, since distance between agents becomes an important factor when determining whether or not to form relationships. Here, I study how heterogeneity in preferences interacts with network effects to determine the aggregate properties of large networks.

To characterize agents' heterogeneity, I classify them into a finite number of types. Let $\Theta$ be the set of agents' types, with $|\Theta| = L < \infty$. For a given number of agents $N$, their types are given by $\theta_1^N, \ldots, \theta_N^N \in \Theta$. Let $\hat{{\bf w}}^N \coloneqq (\hat{w}^N_\theta)_{\theta \in \Theta}$, where $\hat{w}^N_\theta$ is the fraction of the $N$ agents that are of type $\theta$. As $N \to \infty$, let $\hat{{\bf w}}^N \to {\bf w}$ for some ${\bf w} \in \Delta(\Theta)$ with full support over $\Theta$.

In addition to deriving value from network motifs, which only depend on network structure and not on agents' types, I now allow agents' utilities to depend on the distribution of types of their neighbors. This allows us to systematically include heterogeneity in the agents, while also maintaining group-formation incentives. Define the unnormalized empirical distribution of agent $i$'s neighbors in network $g$ as
\begin{align}
    \hat{z}_{i,\theta}^N(g) = \sum_{j=1}^N \mathbbm{1}\{\theta_j^N = \theta\} \mathbbm{1}\{ij \in g\} \quad \forall \theta \in \Theta.
\end{align}
Agents' utility functions now include a component $u: \Theta \times \mathbb{R}^\Theta_+ \to \mathbb{R}$, which depends on their own type and the distribution of their neighbors' types. With this, the utility of agent $i$ in network $g$ is determined by the motif incentives, given by Equation \eqref{eq:motif_utility}, and the value of their neighborhood:
\begin{align}
    U_{i,N}(g) = \underbrace{U^{{\cal M}}_{i,N}(g ; {\bf a})}_{\textrm{motif values}} + \underbrace{u_{\theta_i^N}[\hat{{\bf z}}_i^N(g)]}_{\textrm{neighborhood value}}.
\end{align}
To characterize the limiting distribution of networks, I make some regularity assumptions on $u$.
\begin{assumption} \label{A:neighborhood_utility}
    The neighborhood utility function $u: \Theta \times \mathbb{R}^\Theta_+ \to \mathbb{R}$ has the following properties:
    \begin{itemize}
        \item \textit{Homogeneity}: for all $\theta \in \Theta$, ${\bf z} \in \mathbb{R}^\Theta_+$, and $c > 0$, $u$ satisfies $u_\theta[c {\bf z}] = c u_\theta[{\bf z}]$.
        \item \textit{Concave decomposition}: for all $\theta \in \Theta$, the function $u_\theta$ can be decomposed into
        \begin{align}
            u_\theta({\bf z}) = \sum_{\theta' \in \Theta} c_{\theta \theta'} z_{\theta'} + \Tilde{u}_\theta({\bf z}),
        \end{align}
        where $c_{\theta \theta'} \in \mathbb{R}$ and $\Tilde{u}_\theta$ is continuous, concave and monotonically increasing in all its arguments.
    \end{itemize}
\end{assumption}

These assumptions capture economic environments where agents’ payoffs depend smoothly on the composition of their neighbors. In spatially embedded networks, types may represent locations, so we might want to consider distance-dependent costs. In production or trade networks, types can represent sectors, and the value of a firm’s neighborhood reflects how production scales with inputs (in a CES aggregator, for example), net of costs. Homogeneity ensures that payoffs scale naturally with network size, while concavity and monotonicity guarantee diminishing returns to additional similar neighbors, yielding stable and economically interpretable asymptotic properties.

In the case of homogeneous agents, the distribution resulting from the network formation process is concentrated around networks with some typical density. To extend this result to heterogeneous agents, we need a more flexible characterization of the limiting object. It turns out to be sufficient to characterize the typical density of links between types $\theta$ and $\theta'$. Formally, let ${\cal K}_\Theta$ be the set of functions $\psi: \Theta^2 \to [0,1]$, which I denote as \textit{kernels}. Intuitively, if the process converges to a density kernel $\psi^*$, then $\psi^*_{\theta,\theta'}$ is the fraction of agents of type $\theta'$ that agents of type $\theta$ will typically form a connection with.

Another generalization that must be made to account for agent heterogeneity is of the method of calculating motif densities. Intuitively, the density of motif $m$ in an Erd\H{o}s--R\'enyi network with density $\rho$ is $\rho^{e_m}$. This allows for a clean characterization of the limiting typical networks, as seen in Proposition \ref{prop:motif_partition}. In order to generalize this counting to networks with density kernel $\psi$, the density must now depend on the kernel and the fraction of agents of a given type. Formally, for a kernel $\psi \in {\cal K}_\Theta$ we define the density of motif $m$ as\footnote{For a formal justification for the introduction of this limiting object, see Appendix \ref{sec:app_graph_limits}.}
\begin{align}
    b[m,\psi;{\bf w}] \coloneqq \sum_{\boldsymbol{\theta} \in \Theta^{n_m}} \left( \prod_{i \in {\cal J}_m} w_{\theta_i} \right) \left( \prod_{ij \in m} \psi_{\theta_i \theta_j} \right).
\end{align}
Note that this definition of motif density is equal to $\rho^{e_m}$ for the kernel satisfying $\psi_{\theta \theta'} = \rho$ for all $\theta, \theta'$, which matches our intuition. With this definition, we can study how group formation incentives interact with neighborhood values.
 
\begin{theorem} \label{thm:mult_types_partition}
Suppose that $a_m > 0$ for all motifs with $e_m > 1$ and the neighborhood utility function satisfies Assumption \ref{A:neighborhood_utility}. Define the entropy-adjusted potential of this model as
\begin{align}
    &\Gamma^{\cal M}_\Theta(\psi;{\bf a}, {\bf w}, \sigma) \coloneqq \nonumber \\ 
    &\left( \frac{1-\sigma}{\sigma} \right) \underbrace{\sum_{m \in {\cal M}}  \frac{a_m}{\delta_m} b[m,\psi; {\bf w}]}_{\textrm{motif values}} + \sum_{\theta \in \Theta} w_\theta \Bigg[ \underbrace{\sum_{\theta' \in \Theta} w_{\theta'} H(\psi_{\theta \theta'})}_{\textrm{entropy}} + \left( \frac{1-\sigma}{\sigma} \right) \underbrace{u_\theta[(w_{\theta'} \psi_{\theta \theta'})_{\theta' \in \Theta}]}_{\textrm{neighborhood value}} \Bigg]
\end{align}
Then the limiting scaled normalization constant is given by
\begin{align}
    \zeta^{\cal M}_{\Theta}({\bf a},{\bf w},\sigma) = \max_{\psi \in {\cal K}_\Theta} \Gamma^{\cal M}_\Theta(\psi;{\bf a}, {\bf w}, \sigma).
\end{align}
Furthermore, if the maximizer $\psi^{{\cal M},\Theta}({\bf a},{\bf w},\sigma)$ is unique, then as $N \to \infty$ the networks generated by the model become indistinguishable from those generated by a directed stochastic block model with edge probabilities between types $\theta$ and $\theta'$ given by the kernel $\psi^{{\cal M},\Theta}_{\theta, \theta'}({\bf a},{\bf w})$.
\end{theorem}

Theorem \ref{thm:mult_types_partition} provides a similar reduction in dimensionality as Proposition \ref{prop:motif_partition}, where the complex model of network formation can be understood in terms of a limiting kernel. The same forces drive the result, in the sense that agents must explore an exponentially growing region of network space, so the entropic term appears in the optimization problem. However, now there is heterogeneity in agents' incentives, leading agents of different types to explore in different directions, which results in a potentially non-uniform kernel. In addition to the neighborhood incentives, there is still a global feedback in the form of the motif values. This model provides an additional dimension to the complex interplay between incentives and network formation frictions, in that network structure can now affect the composition of who agents interact with.

\section{A Simple Model of Trade} \label{sec:trade}
In order to build economic intuition on the forces driving the phenomena in Section \ref{sec:large_networks}, I expand the model of trade from Example \ref{ex:simple_trade}. In this setting, assigning values to motifs and neighborhoods arises naturally from the incentives of firms.

\subsection{Technology and incentives}
Consider a set of firms that are spatially embedded in a circular city with unit circumference. There are $L$ equally spaced-out locations where firms can be located. A firm's type is determined by its location, such that $\Theta = \{0,\frac{1}{L}, \ldots, \frac{L-1}{L}\}$. The main quantity that determines the incentives of a firm to trade with another is the distance to the other firm, defined as
\begin{align}
    D(\theta, \theta') = \min\{|\theta-\theta'|,1-|\theta-\theta'|\}.
\end{align}
This definition of distance is illustrated in Figure \ref{fig:distance}. As $N \to \infty$, the fractions of firms of different types converges to ${\bf w}$.

\begin{figure}
    \begin{center}
    \caption{Distance in the trade model}
    \includegraphics[width = 0.5 \linewidth]{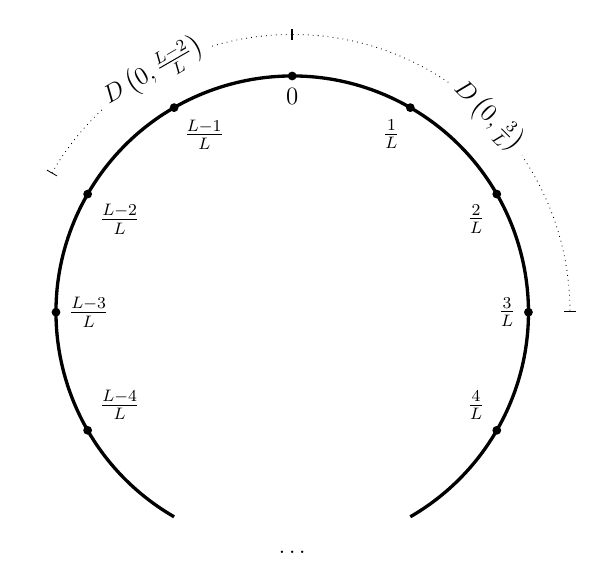}
    \label{fig:distance}
    \end{center}
    \footnotesize{\textit{Notes.} This figure illustrates the distance between firms of different types in the trade model in Section \ref{sec:trade}. The set of types is $\Theta = \{0,\frac{1}{L},\ldots,\frac{L-1}{L}\}$ and the distance is given by $D(\theta, \theta') = \min \{|\theta-\theta'|, 1-|\theta-\theta'|\}$. In this representation, the city has a circumference of 1.}
\end{figure}

Firms play a myopic network formation game, where the presence of a link $ij$ signifies that firm $i$ wishes to establish a trade route with firm $j$. To establish this link, firm $i$ faces a fixed cost of $\gamma D(\theta_i,\theta_j)$, which is proportional to the distance between the firms. If both firms wish to establish a trade route, this route becomes active and generates a revenue of $v$. Hence, for a fixed $N$, the utility of agent $i$ in network $g$ is
\begin{align}
    U_{i,N}(g ; \gamma, v) = v \underbrace{\sum_{j \in {\cal J}_N} \mathbbm{1}\{ij \in g, ji \in g\}}_{\textrm{reciprocated links}} - \gamma \underbrace{\sum_{j \in {\cal J}_N} D(\theta_i,\theta_j) \mathbbm{1}\{ij \in g\}}_{\textrm{outgoing links}}.
\end{align}
This utility can be decomposed into two parts: the value of participating in a mutual-links motif and the cost of establishing trade intentions with a neighborhood. 

In terms of the notation in section \ref{sec:large_networks}, firms assign value to a single motif, such that ${\cal M} = \{m_2\}$, where $m_2 = \{12,21\}$ is the ``mutual links'' motif. The degeneracy of this motif is $\delta_{m_2} = 2$ and its value is $v$. The cost of establishing trade intentions can be written as a neighborhood distribution utility, given by
\begin{align}
    u_\theta[{\bf z}] = - \gamma \sum_{\theta' \in \Theta} z_{\theta'} D(\theta,\theta').
\end{align}
Using this formulation, we can use our previous results on the limiting properties of the network formation game to study the effect of trading incentives on aggregate outcomes.

\subsection{Spatial structure of trade}
Characterizing the network formation game in terms of motif and neighborhood values allows us to formulate the problem of finding the typical density kernel $\psi^{{\cal M}, \Theta}(v, {\bf w}, \sigma)$ that characterizes the limiting properties of the system. Following Theorem \ref{thm:mult_types_partition}, the limiting scaled normalization constant of this model is given by
\begin{align} \label{eq:trade_partition}
    \zeta^{{\cal M}}_\Theta(v, {\bf w},\sigma) = \max_{\psi \in {\cal K}_\Theta} \left[ \sum_{\theta,\theta' \in \Theta^2} w_\theta w_{\theta'} \left( \left( \frac{1-\sigma}{\sigma} \right) \left[ \frac{v}{2} \psi_{\theta \theta'} \psi_{\theta' \theta} -\gamma \psi_{\theta \theta'} D(\theta,\theta') \right] + H(\psi_{\theta \theta'}) \right) \right],
\end{align}
and the typical density kernel $\psi^{{\cal M}, \Theta}(v, {\bf w},\sigma)$ is the maximizer. Interestingly, the solution to this problem can be characterized in terms of our solution to the problem with only motif values.

\begin{lemma} \label{lem:trade_fixed_point}
    Let $\rho^*_{{\cal M}_0}({\bf a},\sigma)$ be the typical density of the model with motifs ${\cal M}_0 = \{m_1,m_2\}$, where $m_1 = \{12\}$ and $m_2 = \{12,21\}$. Suppose the maximizer $\psi^{{\cal M}, \Theta}(v, {\bf w},\sigma)$ in the problem above is unique. Then $\psi^{{\cal M}, \Theta}_{\theta \theta'}(v, {\bf w},\sigma)$ is given by $\rho^*_{{\cal M}_0}(-\gamma D(\theta,\theta'), v, \sigma)$.
\end{lemma}

Intuitively, this result says that in this model there is effectively a network formation process with motif incentives \textit{for each pair} $(\theta,\theta')$, and these do not affect each other for different pairs. This means that our results for phase transitions in the homogeneous agent model hold locally for this model. Therefore, for a fixed $\theta$, $\psi^{{\cal M}, \Theta}_{\theta \theta'}(v, {\bf w})$ can exhibit a large jump as $\theta'$ is varied. 

This characterization allows us to think about the effect of varying the value of gains from trade $v$. This is relevant for policy interventions, such as taxes on trade. In the homogeneous agent model, such a change can trigger a phase transition. However, the heterogeneity in the types of firms can make it such that these effects are greatly mitigated in the total network density.

To build intuition on this result, consider the case where the limiting distribution ${\bf w}$ is the uniform distribution over $\Theta$. The total network density is $\rho = \sum_{(\theta,\theta') \in \Theta^2} w_\theta w_{\theta'} \psi_{\theta \theta'}$. With this, we can study the differential effects of changing $v$ on the total network density and on the density kernel, as shown in Figure \ref{fig:spatial_phase_transition}. In this case, increasing the gains from trade increases the total density of the network in almost imperceptible jumps. In contrast, the density kernel for any two given types can exhibit a stark discontinuous transition. We see that understanding the phase transition structure in the motif-only model gives us insight into the mechanisms of these ``local'' phase transitions as distance changes. Particularly, this model gives rise to clear high-density neighborhoods where firms are trading, and a discontinuous drop in trade with firms outside this neighborhood.

\begin{figure}
    \begin{center}
    \caption{Effect of motif values on network density}%
    \subfloat[Total network density]{{\includegraphics[width=0.49\linewidth]{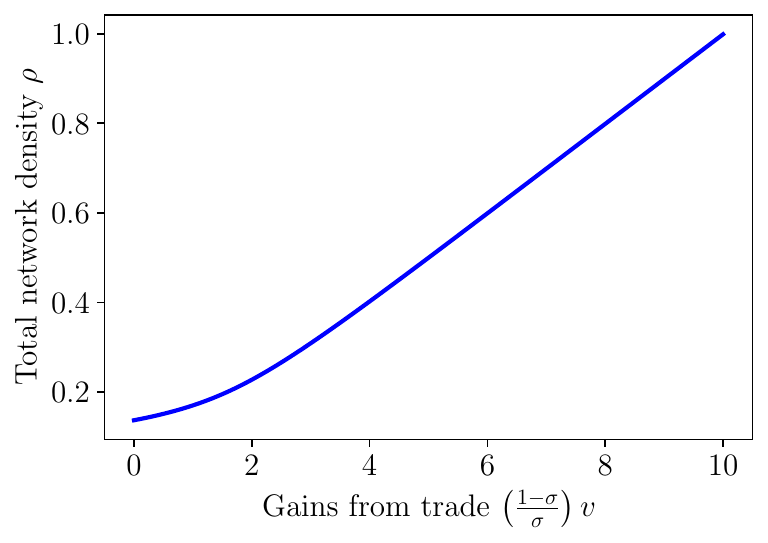} }}%
    \subfloat[Density kernel]{{\includegraphics[width=0.49\linewidth]{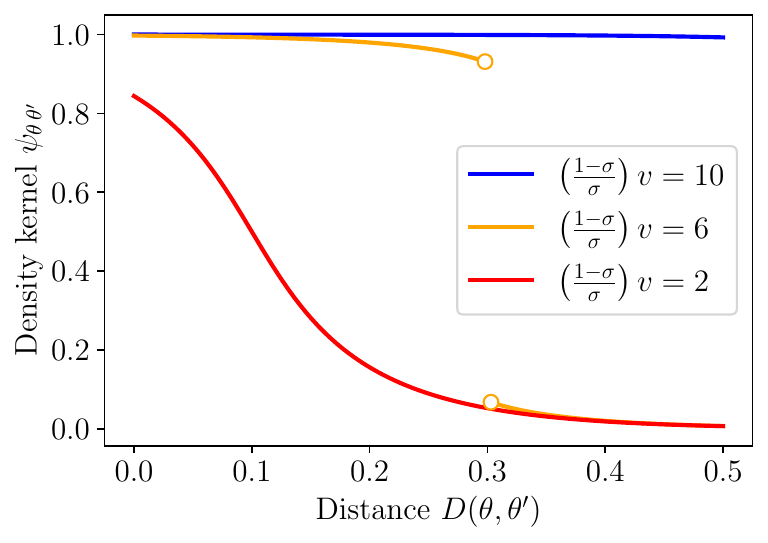} }}
    \label{fig:spatial_phase_transition}
    \end{center}
    \footnotesize{\textit{Notes.} This figure shows the behavior of the typical density and the density kernel as the value of mutual links is changed. There are $L = 100$ types and the cost per unit distance is $\gamma = 10$. Panel (a) shows the total network density as a function of the value of gains from trade. Panel (b) shows the density kernel as a function of distance for different values of $v$.}
\end{figure}
\section{Conclusion}
I have studied a dynamic model of network formation in which agents stochastically meet and myopically decide whether to form or sever links. The paper makes two main contributions: it establishes an alternate formulation of the game in terms of structure values and characterizes the asymptotic behavior of the process as the number of players grows large. When all players in a structure give it the same value, the static game is a potential game and the stationary distribution of the dynamic model can be explicitly computed. This allows us to study the effect of changing the value of these local structures on the aggregate properties of the network. In large populations, the process can be asymptotically approximated by an Erd\H{o}s--R\'enyi process. The typical density of this process can be discontinuous as a function of the values of structures, a phenomenon known as a phase transition.

This work is a stepping stone towards improving our understanding of the economic forces behind the formation of large networks. It gives a clear relation between the incentives to form links and the frictions agents encounter when forming the network, materialized in an entropy term. Despite the simplicity of the motif model, it captures important spillover effects present in real network formation processes. Mainly, forming a structure not only affects current payoffs, but amplifies the probability of forming future links by allowing the formation of even more structures. This self-reinforcing behavior is the driving force behind the phase transitions we observe.

These results open many avenues for further exploration. While a myopic network formation model is a convenient starting point to study aggregate properties, it is natural to ask how the results change when agents are forward-looking. Care is needed, though, since the space of networks grows super-exponentially, so assuming agents can perfectly form expectations over future network realizations might be inappropriate. Another direction to explore is how to extend these results beyond potential games. Potential games have the very desirable property of being equivalent to reversibility (Proposition \ref{prop:conv}), which greatly simplifies the computation of the stationary distribution. However, the large changes at the phase transitions suggest that continuous deviations from potential games (in an appropriate sense) might preserve these stark jumps. I hope this work inspires further exploration in these directions.

\setstretch{1}
\bibliographystyle{aer}
\bibliography{network_formation}
\setstretch{1.5}

\newpage

\appendix

\setcounter{page}{1}
\numberwithin{equation}{section}

\begin{center} \huge \textbf{Appendix} \end{center}
\section{Proofs} \label{sec:app_proofs}
\subsection{Proofs of Lemmas}
\subsubsection{Proof of Lemma \ref{lem:value_representation}}
This is an application of the M\"obius inversion formula for functions of subsets. See Proposition 3.7.2 in \citet{stanley_enumerative_2012} for a proof for a more general setting and Example 3.8.3 for an explicit statement of the result.

\subsubsection{Proof of Lemma \ref{lem:choice_equivalence}}
First, note that choice-equivalence for utility functions $U$ and $\tilde{U}$ holds if and only if
\begin{align} \label{eq:single_link_equivalence}
    U_i(\tau_{ij}(g)) - U_i(g) = \tilde{U}_i(\tau_{ij}(g)) - \tilde{U}_i(g)
\end{align}
for all $g$, $i$ and $ij$. The ``only if'' direction is immediate, since a one-link change is a particular case of a deviation. For the ``if'' direction, notice that any general deviation can be achieved by changing single links, and adding the marginal utility of single-link deviations gives a telescopic sum that results in the marginal utility of the entire deviation.

I now show that the condition \eqref{eq:single_link_equivalence} holds if and only if condition \eqref{eq:value_choice_equivalence} holds. For the rest of the proof, $i$ is fixed. First, fix some network $g$. We have that
\begin{align}
    U_i(g \cup \{ij\}) - U_i(g \backslash \{ij\}) = \sum_{g' \subseteq g \cup \{ij\}} V_i(g') \mathbbm{1}\{ij \in g'\},
\end{align}
and a similar expression involving $\tilde{U}_i$. This means that, regardless of whether $ij \in g$, the marginal utility of changing the link $ij$ only depends on the value of structures that include the link $ij$.\footnote{If $ij \notin g$, then $U_i(\tau_{ij}(g)) - U_i(g)$ is simply the negative of this expression.} In particular, this expression only involves structures $g'$ for which $i \in {\cal J}_{\textrm{src}}(g')$. The ``if'' direction, then, is immediate. If all structures $g'$ for which $i \in {\cal J}_{\textrm{src}}(g')$ satisfy $V_i(g') = \tilde{V}_i(g')$, then condition \eqref{eq:single_link_equivalence} is satisfied.

For the ``only if'' direction, we proceed by induction on the size of the structure.

\textbf{Base case}: Fix $g$ such that $i \in {\cal J}_{\textrm{src}}(g)$ and $|g| = 1$. We can write $g = \{ij\}$. Note that
\begin{align}
    U_i(\tau_{ij}(g)) - U_i(g) = U_i(\varnothing) - U_i(\{ij\}) = -V_i(\{ij\}), \nonumber \\
    \tilde{U}_i(\tau_{ij}(g)) - \tilde{U}_i(g) = \tilde{U}_i(\varnothing) - \tilde{U}_i(\{ij\}) = -\tilde{V}_i(\{ij\}).
\end{align}
Since condition \eqref{eq:single_link_equivalence} holds, we must have $V_i(\{ij\}) = \tilde{V}_i(\{ij\})$. Therefore, all structures with $i \in {\cal J}_{\textrm{src}}(g)$ and $|g| = 1$ satisfy $V_i(g) = \tilde{V}_i(g)$.

\textbf{Induction hypothesis}: For a given $n \ge 1$, assume that all $g$ such that $i \in {\cal J}_{\textrm{src}}(g)$ and $|g| \le n$ satisfy $V_i(g) = \tilde{V}_i(g)$.

\textbf{Inductive step}: Fix $g$ such that $i \in {\cal J}_{\textrm{src}}(g)$ and $|g| = n+1$. Take some $j$ such that $ij \in g$. We have that
\begin{align}
    U_i(\tau_{ij}(g)) - U_i(g) &= -\sum_{g' \subseteq g} V_i(g') \mathbbm{1}\{ij \in g'\}, \nonumber \\
    \tilde{U}_i(\tau_{ij}(g)) - \tilde{U}_i(g) &= -\sum_{g' \subseteq g} \tilde{V}_i(g') \mathbbm{1}\{ij \in g'\}.
\end{align}
Note that all structures $g' \subset g$ satisfy $|g'| \le n$. If $ij \in g'$, then $i \in {\cal J}_{\textrm{src}}(g')$. Therefore, subtracting the two equations above, only the term for the structure $g' = g$ survives (note that $ij \in g$):
\begin{align*}
    [U_i(\tau_{ij}(g)) - U_i(g)] - [\tilde{U}_i(\tau_{ij}(g)) - \tilde{U}_i(g)] = -[V_i(g) - \tilde{V}_i(g)].
\end{align*}
Using condition \eqref{eq:single_link_equivalence}, we obtain $V_i(g) = \tilde{V}_i(g)$. This completes the proof.

\subsubsection{Proof of Lemma \ref{lem:trade_fixed_point}}
The divergence in $H'(u)$ near 0 and 1 ensures the solution to the problem in Equation \eqref{eq:trade_partition} is interior. For $\theta' \ne \theta$, the first-order condition for $\psi_{\theta \theta'}$ is
\begin{align}
    w_\theta w_{\theta'} \left\{ \left( \frac{1-\sigma}{\sigma} \right) \left[ v \psi^*_{\theta' \theta} - \gamma D(\theta,\theta') \right] + H'(\psi^*_{\theta \theta'}) \right\} = 0.
\end{align}
Now, suppose the solution is such that $\psi^*_{\theta \theta'} \ne \psi^*_{\theta' \theta}$. Since $D(\theta,\theta') = D(\theta',\theta)$, making the change $\psi_{\theta \theta'} \leftrightarrow \psi_{\theta' \theta}$ would yield the same value of the objective function. However, this violates the uniqueness of the solution. Therefore, if the solution is unique, we must have $\psi^*_{\theta \theta'} = \psi^*_{\theta' \theta}$.

The argument above implies that $\psi^*_{\theta \theta'}$ satisfies
\begin{align}
    \left( \frac{1-\sigma}{\sigma} \right)(v \psi^*_{\theta \theta'} - \gamma D(\theta, \theta')) + H'(\psi^*_{\theta \theta'}) = 0.
\end{align}
This is precisely the first-order condition of the problem in Equation \eqref{eq:motif_partition} with motif values ${\bf a} = (- \gamma D(\theta,\theta'), v)$.

\subsection{Proof of Proposition \ref{prop:conservative}}
First, note that the game is a potential game with potential $\Phi$ if and only if
\begin{align}
    U_i(\tau_{ij}(g)) - U_i(g) = \Phi(\tau_{ij}(g)) - \Phi(g).
\end{align}
The argument for this is analogous to the argument at the beginning of the proof of Lemma \ref{lem:choice_equivalence}. Intuitively, any deviation can be constructed from one-link deviations, and the marginal utility of an arbitrary deviation is the sum of the marginal utilities of one-link deviations.

To prove the Proposition, it is useful to establish the following technical lemma:
\begin{lemma} \label{lem:conservative_cost}
    Consider a function $\phi:{\cal D}_N \times {\cal G}_N \to \mathbb{R}$ with the property that $\phi_{ij}(g) = - \phi_{ij}(\tau_{ij}(g))$. The function satisfies
    \begin{align} \label{eq:conservative_cost}
        \phi_{ij}(g) + \phi_{i'j'}(\tau_{ij}(g)) = \phi_{i'j'}(g) + \phi_{ij}(\tau_{i'j'}(g))
    \end{align}
    for all $ij, i'j' \in {\cal D}_N$ and $g \in {\cal G}_N$ if and only if there exists a function $\Psi: {\cal G}_N \to \mathbb{R}$ such that
    \begin{align} \label{eq:cost_potential}
        \phi_{ij}(g) = \Psi(\tau_{ij}(g)) - \Psi(g)
    \end{align}
    for all $ij \in {\cal D}_N$ and $g \in {\cal G}_N$. Furthermore, for a given $\phi$, this function $\Psi$ is unique up to an additive constant.
\end{lemma}

\begin{proof}
    The ``if'' direction can be immediately obtained by directly checking Equation \eqref{eq:conservative_cost} using the expression for $\phi$ in terms of $\Psi$.

    The ``only if'' direction can be proven by construction. Fix a network $g \in {\cal G}_N$ and let $n = |g|$ be its size. Let $\gamma = (\gamma_1, \ldots, \gamma_n)$ be an ordering of the links in $g$, meaning that $\gamma_\ell \in g$ for $\ell = 1, \ldots, n$ and $\gamma_\ell \ne \gamma_{\ell'}$ if $\ell \ne \ell'$. From this ordering, we can construct a sequence of networks recursively by letting $g^\gamma_0 \coloneqq \varnothing$ and setting $g^\gamma_{\ell} \coloneqq \tau_{\gamma_\ell}(g^\gamma_{\ell-1})$ for $\ell > 0$. Let us define the following quantity:
    \begin{align}
        \Psi^\gamma(g) \coloneqq \sum_{\ell = 1}^n \phi_{\gamma_\ell}(g^\gamma_{\ell-1}).
    \end{align}
    I now prove that this value does not depend on the specific ordering of the links chosen. 

    First, let us consider an ordering $\gamma^0$ that differs from $\gamma$ only at positions $\ell_0$ and $\ell_0+1$, where $1 \le \ell_0 < n$. Since both $\gamma$ and $\gamma^0$ are orderings, this means that $\gamma^0$ is $\gamma$ with the entries $\ell_0$ and $\ell_0+1$ swapped. From this, we must have $g^{\gamma^0}_{\ell_0-1} = g^{\gamma}_{\ell_0-1}$, and $g^{\gamma^0}_{\ell_0+2} = g^{\gamma}_{\ell_0+2}$ if $\ell_0 < n-1$. Therefore, computing $\Psi^\gamma(g) - \Psi^{\gamma^0}(g)$ yields
    \begin{align}
        \Psi^\gamma(g) - \Psi^{\gamma^0}(g) = \phi_{\gamma_{\ell_0}}(g^\gamma_{\ell_0-1}) + \phi_{\gamma_{\ell_0+1}}(g^\gamma_{\ell_0}) - \phi_{\gamma^0_{\ell_0}}(g^{\gamma^0}_{\ell_0-1}) - \phi_{\gamma^0_{\ell_0+1}}(g^{\gamma^0}_{\ell_0}).
    \end{align}
    Note that we can write $g^{\gamma^0}_{\ell_0} = \tau_{\gamma_{\ell_0+1}}(g^{\gamma}_{\ell_0 - 1})$. Therefore, this difference can be written as
    \begin{align}
        \Psi^\gamma(g) - \Psi^{\gamma^0}(g) = \phi_{\gamma_{\ell_0}}(g^\gamma_{\ell_0-1}) + \phi_{\gamma_{\ell_0+1}}(\tau_{\gamma_{\ell_0}}(g^\gamma_{\ell_0-1})) - \phi_{\gamma_{\ell_0+1}}(g^{\gamma}_{\ell_0-1}) - \phi_{\gamma_{\ell_0}}(\tau_{\gamma_{\ell_0+1}}(g^{\gamma}_{\ell_0 - 1})).
    \end{align}
    From the condition in Equation \eqref{eq:conservative_cost}, this must vanish, such that $\Psi^\gamma(g) = \Psi^{\gamma^0}(g)$. Any ordering of the links in $g$ can be obtained starting from $\gamma$ through a sequence of ``swaps'' like the one we just analyzed. Therefore, the quantity $\Psi^\gamma(g)$ will be the same regardless of the ordering $\gamma$ chosen. Thus, we can define a quantity that depends solely on the network $\Psi(g) \coloneqq \Psi^\gamma(g)$ for any ordering $\gamma$.

    Now we want to prove that our construction of $\Psi$ satisfies Equation \eqref{eq:cost_potential}. Fix a network $g$ and a link $ij \in {\cal D}_N$. Without loss of generality, suppose that $ij \notin g$\footnote{Since $\phi_{ij}(g) = - \phi_{ij}(\tau_{ij}(g))$, it suffices to prove that Equation \eqref{eq:cost_potential} holds for $g$ to prove that it holds for $\tau_{ij}(g)$, so we can choose to prove it for the smaller network.}. Let $\gamma^1$ be some ordering of the links in $g$ and $\gamma^2$ be the ordering of the links in $\tau_{ij}(g)$ such that $\gamma^2_\ell = \gamma^1_\ell$ for $1 \le \ell \le |g|$ and $\gamma^2_{|g|+1} = ij$. We can write our $\Psi$ functions using these orderings, such that
    \begin{align}
        \Psi(\tau_{ij}(g)) - \Psi(g) = \sum_{\ell = 1}^{|g|+1} \phi_{\gamma^2_\ell}(g^{\gamma^2}_{\ell-1}) - \sum_{\ell = 1}^{|g|} \phi_{\gamma^1_\ell}(g^{\gamma^1}_{\ell-1}) = \phi_{ij}(g),
    \end{align}
    as desired.

    Finally, to prove uniqueness up to an additive constant, suppose that there are two functions $\Psi_1$ and $\Psi_2$ satisfying Equation \eqref{eq:cost_potential} for all $g \in {\cal G}_N$ and $ij \in {\cal D}_N$. Let $K_0 \coloneqq \Psi_1(\varnothing) - \Psi_2(\varnothing)$, and suppose there exists a $g_0$ such that $\Psi_1(g_0) - \Psi_2(g_0) \ne K_0$. From Equation \eqref{eq:cost_potential}, we have that, for all $g \in {\cal G}_N$ and $ij \in {\cal D}_N$
    \begin{align}
        \Psi_1(g) - \Psi_2(g) &= [\Psi_1(g) + \phi_{ij}(g)] - [\Psi_2(g) + \phi_{ij}(g)] \nonumber \\
        &= \Psi_1(\tau_{ij}(g)) - \Psi_2(\tau_{ij}(g)).
    \end{align}
    Since this procedure can be repeated with any sequence of dyads to reach any network, we conclude that $\Psi_1(g) - \Psi_2(g) = K_0$ for all networks, which contradicts our assumption. 
\end{proof}

We can now prove Proposition \ref{prop:conservative}. To do this, define $\phi_{ij}(g) \coloneqq U_i(\tau_{ij}(g)) - U_i(g)$. Conservativeness, then, is equivalent to Equation \eqref{eq:conservative_cost}. Lemma \ref{lem:conservative_cost}, then, establishes that conservativeness is equivalent to the existence of a potential, and that this potential is unique up to an additive constant.

\subsection{Proof of Proposition \ref{prop:value_representation}}
Let $V$ be the structure values associated to the utility functions $U$. Define the utilities $\tilde{U}$ by
\begin{align}
    \tilde{U}_i(g) \coloneqq \sum_{g' \subseteq g} V_i(g') \mathbbm{1}\{i \in {\cal J}_{\textrm{src}}(g')\},
\end{align}
with corresponding values
\begin{align}
    \tilde{V}_i(g) \coloneqq V_i(g) \mathbbm{1}\{i \in {\cal J}_{\textrm{src}}(g)\}.
\end{align}
By Lemma \ref{lem:choice_equivalence}, these utilities are choice-equivalent to $U$. Therefore, to prove Proposition \ref{prop:value_representation} it is sufficient to prove that $\tilde{V}$ satisfies condition \eqref{eq:value_potential_cond} if and only if the game is a potential game.

First, suppose the game is a potential game with potential $\Phi$. Fix a network $g = (s_i, g_{-i})$. Note that
\begin{align}
    \Phi(s_i,g_{-i}) - \Phi(\varnothing,g_{-i}) = \tilde{U}_i(s_i,g_{-i}) - \tilde{U}_i(\varnothing,g_{-i}) = \tilde{U}_i(s_i,g_{-i}).
\end{align}
The last equality comes from the fact that $i \notin {\cal J}_{\textrm{src}}(\varnothing, g_{-i})$, so the payoff associated to it is 0. To simplify notation, denote with $\hat{g}_{-i}$ the network $(\varnothing, g_{-i})$. Using our inversion formula from Lemma \ref{lem:value_representation} applied to the potential, define $V_\Phi$ as the unique values such that
\begin{align}
    \Phi(g) \coloneqq \sum_{g' \subseteq g} V_\Phi(g').
\end{align}
Now, note that we can write the utility functions as
\begin{align}
    \tilde{U}_i(g) = \Phi(g) - \Phi(\hat{g}_{-i}) = \sum_{g' \subseteq g} V_\Phi(g') \mathbbm{1}\{g' \not\subseteq \hat{g}_{-i}\}.
\end{align}
Note that a structure $g' \subseteq g$ satisfies $g' \not\subseteq \hat{g}_{-i}$ if and only if $i \in {\cal J}_{\textrm{src}}(g')$. This means that
\begin{align}
    \tilde{U}_i(g) = \sum_{g' \subseteq g} V_\Phi(g') \mathbbm{1}\{i \in {\cal J}_{\textrm{src}}(g')\}.
\end{align}
Using Lemma \ref{lem:value_representation}, we know that the structure values for a set of utilities are unique. This implies that
\begin{align}
    V_\Phi(g) \mathbbm{1}\{i \in {\cal J}_{\textrm{src}}(g)\} = \tilde{V}_i(g) = V_i(g) \mathbbm{1}\{i \in {\cal J}_{\textrm{src}}(g)\}.
\end{align}
Therefore, $V_i(g) = V_\Phi(g)$ for all $i \in {\cal J}_{\textrm{src}}(g)$.

Now suppose that $V_i(g) = V_0(g)$ for all $i \in {\cal J}_{\textrm{src}}(g)$ and let $V_0(\varnothing)$ be arbitrary. I will now prove that the game is a potential game. We can write the utilities $\tilde{U}$ as
\begin{align}
    \tilde{U}_i(g) = \sum_{g' \subseteq g} V_0(g') \mathbbm{1}\{i \in {\cal J}_{\textrm{src}}(g')\}.
\end{align}
We can prove that the game is a potential game by proposing a potential and checking that it captures the marginal utilities of deviating. Define the proposed potential as
\begin{align}
    \Phi_V(g) \coloneqq \sum_{g' \subseteq g} V_0(g').
\end{align}
As shown in the proof of Proposition \ref{prop:conservative}, it is sufficient to show that the potential captures the incentives of a single-link deviation. First, fix $g$ such that $ij \in g$. Then we have
\begin{align}
    \tilde{U}_i(\tau_{ij}(g)) - \Tilde{U}_i(g) = -\sum_{g' \subseteq g} V_0(g') \mathbbm{1}\{i \in {\cal J}_{\textrm{src}}(g'), ij \in g'\}.
\end{align}
Note that $ij \in g' \implies i \in {\cal J}_{\textrm{src}}(g')$. Therefore, we can write this marginal utility as
\begin{align}
    \tilde{U}_i(\tau_{ij}(g)) - \Tilde{U}_i(g) = -\sum_{g' \subseteq g} V_0(g') \mathbbm{1}\{ij \in g'\}.
\end{align}
Similarly, we can compute the change in the potential from the deviation
\begin{align}
    \Phi_V(\tau_{ij}(g)) - \Phi_V(g) = -\sum_{g' \subseteq g} V_0(g') \mathbbm{1}\{ij \in g'\} = \tilde{U}_i(\tau_{ij}(g)) - \Tilde{U}_i(g).
\end{align}
If $ij \notin g$, define $g' \coloneqq \tau_{ij}(g)$ and apply the steps above to reach the same result. Therefore, the game is a potential game with potential $\Phi_V$.

\subsection{Proof of Proposition \ref{prop:conv}}
First, let us prove that the existence of the potential implies reversibility of the process. By definition, the potential satisfies
\begin{align}
    \Phi(\tau_{ij}(g))-\Phi(g) = U_i(\tau_{ij}(g)) - U_i(g) = \left( \frac{\sigma}{1-\sigma} \right) \log \left( \frac{p_{ij}(g)}{p_{ij}(\tau_{ij}(g))} \right).
\end{align}
The Markov process is reversible if there exists a distribution $\pi$ such that the detailed balance condition is satisfied:
\begin{align}
    \lambda_{ij}(g) p_{ij}(g) \pi(g) = \lambda_{ij}(\tau_{ij}(g)) p_{ij}(\tau_{ij}(g)) \pi(\tau_{ij}(g))
\end{align}
for all $g \in {\cal G}_N$ and $ij \in {\cal D}_N$. This is sufficient since the only non-zero transition probabilities correspond to pairs of networks that differ by only one link. Since $\lambda_{ij}(\tau_{ij}(g)) = \lambda_{ij}(g)$, we have that 
\begin{align}
    \lambda_{ij}(g) p_{ij}(g) e^{\left( \frac{1-\sigma}{\sigma} \right) \Phi(g)} &= \lambda_{ij}(\tau_{ij}(g)) p_{ij}(\tau_{ij}(g)) \left[ \frac{p_{ij}(g)}{p_{ij}(\tau_{ij}(g))} \right] e^{\left( \frac{1-\sigma}{\sigma} \right) \Phi(g)} \nonumber \\
    &= \lambda_{ij}(\tau_{ij}(g)) p_{ij}(\tau_{ij}(g)) e^{\left( \frac{1-\sigma}{\sigma} \right)[\Phi(\tau_{ij}(g))-\Phi(g)]} e^{\left( \frac{1-\sigma}{\sigma} \right)\Phi(g)} \nonumber \\
    &= \lambda_{ij}(\tau_{ij}(g)) p_{ij}(\tau_{ij}(g)) e^{\left( \frac{1-\sigma}{\sigma} \right) \Phi(\tau_{ij}(g))}.
\end{align}
Define the following probability distribution over the space of networks: 
\begin{align} \label{eq:stationary_dist}
    \Tilde{\pi}(g) = \frac{\exp\left[\left( \frac{1-\sigma}{\sigma} \right) \Phi(g)\right]}{\sum_{g' \in {\cal G}_N} \exp\left[\left( \frac{1-\sigma}{\sigma} \right) \Phi(g')\right]}.
\end{align}
Dividing the previous expression by $\sum_{g' \in {\cal G}_N} \exp\left[\left( \frac{1-\sigma}{\sigma} \right) \Phi(g')\right]$, we have that
\begin{align}
    \lambda_{ij}(g) p_{ij}(g) \tilde{\pi}(g) = \lambda_{ij}(\tau_{ij}(g)) p_{ij}(\tau_{ij}(g)) \tilde{\pi}(\tau_{ij}(g)),
\end{align}
so the process is reversible. In addition to Assumption \ref{A:T1EV} (from which we have irreducibility and aperiodicity), satisfying detailed balance implies that $\tilde{\pi}$ is the unique stationary distribution of the Markov chain.

Now we can prove that reversibility of the process implies existence of the potential. Recall that under Assumption \ref{A:T1EV}, there exists a unique non-degenerate stationary distribution $\pi$ to which the system converges. If the process is reversible, it satisfies the detailed balance condition:
\begin{align}
    \lambda_{ij}(g) p_{ij}(g) \pi(g) &= \lambda_{ij}(\tau_{ij}(g)) p_{ij}(\tau_{ij}(g)) \pi(\tau_{ij}(g))
\end{align}
Since the stationary distribution has full support and $\lambda_{ij}(g) = \lambda_{ij}(\tau_{ij}(g))$, we have that
\begin{align}
    \log(\pi(\tau_{ij}(g))) - \log(\pi(g)) &= \log \left( \frac{p_{ij}(g)}{p_{ij}(\tau_{ij}(g))} \right).
\end{align}
for all $g \in {\cal G}_N$ and $ij \in {\cal D}_N$. Since the stationary probabilities are a function of the network state only, we have that any function
\begin{align}
    \Phi(g) = C + \left( \frac{\sigma}{1-\sigma} \right) \log(\pi(g)),
\end{align}
with $C \in \mathbb{R}$, will satisfy
\begin{align}
    \Phi(\tau_{ij}(g)) - \Phi(g) = \left( \frac{\sigma}{1-\sigma} \right) \log\left( \frac{p_{ij}(g)}{p_{ij}(\tau_{ij}(g))} \right) = U_i(\tau_{ij}(g)) - U_i(g),
\end{align}
so the static game is a potential game.

Using the proof of the first implication, we have that the potential characterizes the stationary distribution of the chain. Namely, the stationary distribution is given by Equation \eqref{eq:stationary_dist}.

\subsection{Proof of Proposition \ref{prop:motif_partition}}
For definitions and important results, see Appendix \ref{sec:app_graph_limits}.

First, we want to find the set of graphons that solve the optimization problem in Theorem \ref{thm:partition_func_convergence}. Note that in our case, the function ${\cal T}$ corresponds to our scaled potential. These technically correspond to subgraph densities but, as pointed out in Appendix \ref{sec:app_graph_limits}, homomorphism densities and subgraph densities are equivalent for large networks. Specifically, in the proofs of Theorem \ref{thm:partition_func_convergence} and \ref{thm:graphon_convergence}, a subleading correction will be removed by the limsup and liminf. Therefore, we have that for a given graphon $h \in {\cal W}$,
\begin{align}
    {\cal T}(h) = \sum_{m \in {\cal M}} \frac{a_m}{\delta_m} \tilde{b}(m,h),
\end{align}
where $\tilde{b}(m,h)$ is the homomorphism density of motif $m$ into graphon $h$. To solve this problem, I will first prove that this optimization problem is solved by a constant graphon. 

Fix some graphon $h \in {\cal W}$ and a motif $m \in {\cal M}$. For our case of interest, H\"older's inequality states that for two functions $f,g:[0,1]^n \to [0,1]$ and for $p,q > 0$ such that $1/p +  1/q = 1$, we have that
\begin{align}
    \int_{[0,1]^n} f({\bf x}) g({\bf x}) \, dx_1 \ldots dx_n \le \left( \int_{[0,1]^n} f({\bf x})^p \, dx_1 \ldots dx_n \right)^{1/p} \left( \int_{[0,1]^n} g({\bf x})^q \, dx_1 \ldots dx_n \right)^{1/q},
\end{align}
with equality if and only if there exist $c_1, c_2 \in \mathbb{R}$, not both zero, such that $c_1 f({\bf x})^p = c_2 g({\bf x})^q$ almost everywhere. Applying this repeatedly with appropriate choices of $p$ and $q$, we have that for a collection of $k$ functions $f_i:[0,1]^n \to [0,1]$, the following holds:
\begin{align}
    \int_{[0,1]^n} \left( \prod_{i=1}^k f_i({\bf x}) \right) \, dx_1 \ldots dx_n \le \prod_{i=1}^k \left( \int_{[0,1]^n} f_i({\bf x})^k \, dx_1 \ldots dx_n \right)^{1/k}.
\end{align}
Applying this result to the homomorphism density of $m$ in $h$, we obtain
\begin{align}
    \tilde{b}(m,h) &= \int_{[0,1]^{n_m}} \left( \prod_{ij \in m} h(x_i, x_j) \right) dx_1 \ldots dx_{n_m} \nonumber \\
    &\le \prod_{ij \in m} \left( \int_{[0,1]^{n_m}} h(x_i, x_j)^{e_m} \, dx_1 \ldots dx_{n_m} \right)^{1/e_m} \nonumber \\
    &= \int_{[0,1]^2} h(x, y)^{e_m} \, dx \, dy,
\end{align}
where $e_m = |m|$ is the number of edges in the motif $m$. Additionally, note that the above is always an equality if $e_m = 1$.

Recall that we assumed $a_m > 0$ for all $m$ with $e_m > 1$. We can apply the previous result to the motifs with $e_m > 1$ to bound our function ${\cal T}$ for an arbitrary graphon $h$:
\begin{align}
    {\cal T}(h) \le \sum_{m \in {\cal M}} \frac{a_m}{\delta_m} \int_{[0,1]^2} h(x, y)^{e_m} \, dx \, dy.
\end{align}
Define the function $R:[0,1] \to \mathbb{R}$ as
\begin{align}
    R(s) \coloneqq \sum_{m \in {\cal M}} \frac{a_m}{\delta_m} s^{e_m} + H(s).
\end{align}
Let ${\cal H}(h) \coloneqq \int_{[0,1]^2} H(h(x,y)) \, dx \, dy$. From the bound above, we see that
\begin{align}
    {\cal T}(h) + {\cal H}(h) \le \int_{[0,1]^2} R(h(x,y)) \, dx \, dy.
\end{align}
Let $R^* \coloneqq \sup_{s} R(s)$. Using the extreme value theorem, we know this supremum is actually a maximum, and is achieved at some $s^*$. Therefore, we have that
\begin{align}
    {\cal T}(h) + {\cal H}(h) \le \int_{[0,1]^2} R(s^*) \, dx \, dy.
\end{align}
Now, the graphon $h^*$ given by $h^*(x,y) = s^*$ saturates H\"older's inequality and, therefore, satisfies
\begin{align}
    {\cal T}(h^*) + {\cal H}(h^*) = \int_{[0,1]^2} R(s^*) \, dx \, dy,
\end{align}
so it solves the optimization problem in Theorem \ref{thm:partition_func_convergence}.

Uniqueness of the solution follows the same argument as Theorem 4.1 in \citet{chatterjee_estimating_2013}. Therefore, convergence to the Erd\H{o}s--R\'enyi model immediately follows from Theorem \ref{thm:graphon_convergence}.

\subsection{Proof of Theorem \ref{thm:mult_types_partition}}
To prove Theorem \ref{thm:mult_types_partition}, we want to find the colored graphon that solves the optimization problem in Theorem \ref{thm:colored_graphon_partition} in Appendix \ref{sec:app_graph_limits}. Intuitively, we want to prove that this problem is solved by a graphon that is piece-wise constant. 

Our potential is now composed of two terms ${\cal T}_1$ and ${\cal T}_2$, corresponding to the motif and neighbor utilities, respectively. These can be written as
\begin{align}
    {\cal T}_1(h,c) &= \sum_{m \in {\cal M}} \frac{a_m}{\delta_m} \tilde{b}(m,h), \nonumber \\
    {\cal T}_2(h,c) &= \sum_{\theta \in \Theta} \int_{[0,1]} \mathbbm{1}\{c(x)=\theta\} u_\theta \left( \left[ \int_{[0,1]} h(x,y) \mathbbm{1}\{c(y)=\theta'\} \, dy \right]_{\theta' \in \Theta} \right) \, dx.
\end{align}
To simplify notation, it is convenient to define ``subgraphons'' for every pair of types. Without loss of generality, we can consider the coloring to be ordered\footnote{This is because ${\cal T}$ is the same for all measurable re-labelings of the nodes.}, such that for some ordering of types $\theta^1, \ldots, \theta^L$, the coloring is
\begin{align}
    c(x) = \sum_{\theta^i \in \Theta} \theta^i \mathbbm{1}\left\{ x \ge \sum_{j < i} w_{\theta^j}, x < \sum_{j \le i} w_{\theta^j} \right\}.
\end{align}
The subgraphon associated to types $\theta^i$ and $\theta^j$ is, then, $h_{\theta^i \theta^j}: [0,1]^2 \to [0,1]$ defined by
\begin{align}
    h_{\theta^i \theta^j}(x,y) \coloneqq h\left( \sum_{k < i} w_{\theta^k} + w_{\theta^i} x, \sum_{k < j} w_{\theta^k} + w_{\theta^j} y \right).
\end{align}
This will allow us to consider the variational problem on the whole graphon as a variational problem on the subgraphons $h_{\theta \theta'}$.

To begin, we fix a colored graphon $(h,c)$ and find a bound on ${\cal T}_1(h)$ in a similar manner to before. For a given motif $m$, we can break up the integration domain into ``boxes'' where all the types are fixed. Specifically, we can write
\begin{align}
    \tilde{b}(m,h) &= \sum_{\boldsymbol{\theta} \in \Theta^{n_m}} \int_{[0,1]^{n_m}} \mathbbm{1}\{c(x_k) = \theta_k \, \forall k \} \left( \prod_{ij \in m} h(x_i, x_j) \right) \, dx_1 \ldots dx_{n_m} \nonumber \\
    &= \sum_{\boldsymbol{\theta} \in \Theta^{n_m}} \left( \prod_{i \in {\cal N}_m} w_{\theta_i} \right) \int_{[0,1]^{n_m}} \left( \prod_{ij \in m} h_{\theta_i \theta_j}(x_i, x_j) \right) \, dx_1 \ldots dx_{n_m}.
\end{align}
We will now use the following result, obtained by iterating H\"older's inequality:
\begin{align}
    \int_{[0,1]^{n_m}} \left( \prod_{ij \in m} f_{ij}(x_i, x_j) \right) dx_1 \ldots dx_{n_m} \le \prod_{ij \in m} \lVert f_{ij} \rVert_{e_m}.
\end{align}
Applying this to the homomorphism density, we obtain
\begin{align}
    \tilde{b}(m,h) &\le \sum_{\boldsymbol{\theta} \in \Theta^{n_m}} \left( \prod_{i \in {\cal N}_m} w_{\theta_i} \right) \prod_{ij \in m} \lVert h_{\theta_i \theta_j} \rVert_{e_m}.
\end{align}
We can use this to bound ${\cal T}_1$. Note that we can assume without loss of generality that $e_m > 1$ for all motifs $m$, since any motifs with $e_m = 1$ can be absorbed into the linear part of the neighborhood utility. Recall that we assumed that $a_m > 0$ for all motifs with $e_m > 1$. Then, we have that $a_m > 0$ for all $m$, so we obtain the bound
\begin{align}
    {\cal T}_1(h,c) \le \sum_{m \in {\cal M}} \frac{a_m}{\delta_m} \sum_{\boldsymbol{\theta} \in \Theta^{n_m}} \left( \prod_{i \in {\cal N}_m} w_{\theta_i} \right) \prod_{ij \in m} \lVert h_{\theta_i \theta_j} \rVert_{e_m}.
\end{align}
Note that this bound has different norms based on the size of the motif. To obtain a bound that uses a single norm for all graphons, we can use the fact that for a space with $\lVert 1 \rVert_1 = 1$, for $1 \le p \le q < \infty$, the corresponding norms satisfy $\lVert h \rVert_p \le \lVert h \rVert_q$, with equality iff $h$ is constant almost everywhere. Define $e^* \coloneqq \max_{m \in {\cal M}} e_m$. Then we have the following bound:
\begin{align}
    {\cal T}_1(h,c) \le \sum_{m \in {\cal M}} \frac{a_m}{\delta_m} \sum_{\boldsymbol{\theta} \in \Theta^{n_m}} \left( \prod_{i \in {\cal N}_m} w_{\theta_i} \right) \prod_{ij \in m} \lVert h_{\theta_i \theta_j} \rVert_{e^*}.
\end{align}

To find a bound on ${\cal T}_2$, I will use the following technical Lemma.

\begin{lemma} \label{lem:entropy_bound}
    For any $n \ge 1$, $h \in {\cal W}$ and $\kappa \in \mathbb{R}$,
    \begin{align}
        \int_{[0,1]^2} [H(h(x,y)) + \kappa h(x,y)] \, dx \, dy \le H(\lVert h \rVert_{n}) + \kappa \lVert h \rVert_{n}.
    \end{align}
\end{lemma}
\begin{proof}
    Note that the function $G(u) \coloneqq H(u) + \kappa u$ has a unique maximizer, which I call $\rho^*(\kappa)$. I consider the cases $\lVert h \rVert_{n} \le \rho^*(\kappa)$ and $\lVert h \rVert_{n} > \rho^*(\kappa)$ separately.

    First, consider the case $\lVert h \rVert_{n} \le \rho^*(\kappa)$. Since $G$ is a concave function, we can use Jensen's inequality to obtain
    \begin{align}
        \int_{[0,1]^2} G(h(x,y)) \, dx \, dy \le G(\lVert h \rVert_1).
    \end{align}
    Since $G(u)$ is non-decreasing for $u \le \rho^*(\kappa)$ and $\lVert h \rVert_1 \le \lVert h \rVert_n$, we obtain the bound for this case.

    For the case $\lVert h \rVert_n > \rho^*(\kappa)$, I show that the optimization problem
    \begin{align}
        \max_{h' \in {\cal W}} \int [H(h'(x,y)) + \kappa h'(x,y)] \, dx \, dy \quad \textrm{s.t.} \quad \lVert h' \rVert_{n} \ge \rho,
    \end{align}
    for $\rho \coloneqq \lVert h \rVert_n > \rho^*(\kappa)$, is solved by the constant graphon $h'(x,y) = \rho$. Note that $G$ satisfies
    \begin{align}
        \lim_{u \to 0} G'(u) = +\infty, \quad \lim_{u \to 1} G'(u) = -\infty.
    \end{align}
    This means that the constraints $h'(x,y) \in [0,1]$ will not be binding. Using Theorem 9.4.1 in \cite{luenberger_optimization_1998}, we have that the solution $h_n$ must be a stationary point of the Langrangian
    \begin{align}
        {\cal L}(h') \coloneqq \int [H(h'(x,y)) + \kappa h'(x,y) + \mu (\rho^n - h'(x,y)^n)] \, dx \, dy,
    \end{align}
    for some Lagrange multiplier $\mu$. This means that it must satisfy
    \begin{align}
        H'(h_n(x,y)) + \kappa - n \mu h_n(x,y)^{n-1} = 0,
    \end{align}
    with $\mu \ge 0$ and the complementary slackness condition
    \begin{align}
        \mu \left[ \rho^n - \int_{[0,1]^2} h_n(x,y)^n \, dx \, dy \right] = 0.
    \end{align}
    Note that we must have $\mu > 0$, since $H'(h_n) + \kappa = 0$ would imply $h_n = \rho^*(\kappa) \implies \lVert h_n \rVert_n = \rho^*(\kappa)$, which violates the constraint. With $\mu > 0$, the function $H'(u) + \kappa - n \mu u^{n-1}$ is strictly decreasing, and diverges as it approaches $0$ and $1$, so it must have a unique root. This means that there is a unique solution to the problem where the constraint binds, corresponding to $h_n = \rho$. Since our original graphon is in the feasible set of this problem, it is bounded by the value of the functional at the constant graphon. This yields the bound for the second case.
\end{proof}

Let us write ${\cal T}_2$ in terms of the decomposition of $u$:
\begin{align}
    {\cal T}_2(h,c) = \sum_{\theta \in \Theta} w_\theta \int_{[0,1]} \left[ \sum_{\theta' \in \Theta} c_{\theta \theta'} w_{\theta'} \int_{[0,1]} h_{\theta \theta'}(x,y) \, dy + \tilde{u}_\theta\left( \left[ w_{\theta'} \int_{[0,1]} h_{\theta \theta'}(x,y) \, dy \right]_{\theta' \in \Theta} \right) \right] \, dx.
\end{align}
In addition to this, we can also write the entropy term in terms of the subgraphons:
\begin{align}
    {\cal H}[h] = \sum_{\theta,\theta' \in \Theta} w_\theta w_{\theta'} \int_{[0,1]^2} H(h_{\theta,\theta'}(x,y)) \, dx \,dy.
\end{align}
Note that the Lemma \ref{lem:entropy_bound} allows us to jointly bound ${\cal T}_2 + {\cal H}$. Specifically, we obtain
\begin{align}
    {\cal T}_2(h,c) + {\cal H}(h) \le & \sum_{\theta,\theta' \in \Theta} w_\theta w_{\theta'} [H(\lVert h_{\theta \theta'} \rVert_{e^*}) + c_{\theta \theta'} \lVert h_{\theta \theta'} \rVert_{e^*}] \nonumber \\
    &+ \sum_{\theta \in \Theta} w_\theta \int_{[0,1]} \tilde{u}_\theta\left( \left[ w_{\theta'} \int_{[0,1]} h_{\theta \theta'}(x,y) \, dy \right]_{\theta' \in \Theta} \right) \, dx.
\end{align}
Using concavity of the $\tilde{u}_\theta$ functions, we can apply Jensen's inequality to obtain
\begin{align}
    \int_{[0,1]} \tilde{u}_\theta\left( \left[ w_{\theta'} \int_{[0,1]} h_{\theta \theta'}(x,y) \, dy \right]_{\theta' \in \Theta} \right) \le \tilde{u}_\theta\left( \left[ w_{\theta'} \int_{[0,1]^2} h_{\theta \theta'}(x,y) \, dx \, dy \right]_{\theta' \in \Theta} \right) = \tilde{u}_\theta\left( \left[ w_{\theta'} \lVert h_{\theta \theta'} \rVert_1 \right]_{\theta' \in \Theta} \right).
\end{align}
Since $\lVert h_{\theta \theta'} \rVert_1 \le \lVert h_{\theta \theta'} \rVert_{e^*}$, using monotonicity of $\tilde{u}_\theta$ yields
\begin{align}
    \int_{[0,1]} \tilde{u}_\theta\left( \left[ w_{\theta'} \int_{[0,1]} h_{\theta \theta'}(x,y) \, dy \right]_{\theta' \in \Theta} \right) \le \tilde{u}_\theta\left( \left[ w_{\theta'} \lVert h_{\theta \theta'} \rVert_{e^*} \right]_{\theta' \in \Theta} \right).
\end{align}
Together with the previous bound, this yields
\begin{align}
    {\cal T}_2(h,c) + {\cal H}(h) \le \sum_{\theta \in \Theta} w_\theta \left[ \sum_{\theta' \in \Theta} w_{\theta'} [H(\lVert h_{\theta \theta'} \rVert_{e^*}) + u_\theta([ w_{\theta'} \lVert h_{\theta \theta'} \rVert_{e^*} ]_{\theta' \in \Theta} ) \right].
\end{align}

As in the main text, let ${\cal K}_\Theta$ be the set of functions $\psi:\Theta^2 \to [0,1]$. For $\psi \in {\cal K}_\Theta$, define the function
\begin{align}
    Q(\psi) \coloneqq \sum_{m \in {\cal M}} \frac{a_m}{\delta_m} \sum_{\boldsymbol{\theta} \in \Theta^{n_m}} \left( \prod_{i \in {\cal N}_m} w_{\theta_i} \right) \prod_{ij \in m} \psi_{\theta_i \theta_j} + \sum_{\theta \in \Theta} w_\theta \left[ \sum_{\theta' \in \Theta} w_{\theta'} [H(\psi_{\theta \theta'}) + u_\theta([w_{\theta'} \psi_{\theta \theta'} ]_{\theta' \in \Theta} ) \right].
\end{align}
Then the bounds above imply that, for an arbitrary colored graphon $(h,c)$,
\begin{align}
    {\cal T}(h,c) + {\cal H}(h) \le Q([\lVert h_{\theta \theta'} \rVert_{e^*}]_{\theta,\theta' \in \Theta}).
\end{align}
Define $\psi^*$ to be
\begin{align}
    \psi^* \in \argmax_{\psi \in {\cal K}_\Theta} \, Q(\psi).
\end{align}
This optimality implies that any graphon satisfies the following bound:
\begin{align}
    {\cal T}(h,c) + {\cal H}(h) \le Q(\psi^*).
\end{align}
Let $h^*$ be the piecewise constant graphon given by $h^*_{\theta \theta'}(x,y) = \psi^*_{\theta \theta'}$ for all $\theta, \theta', x, y$. Note that this graphon satisfies $\lVert h^*_{\theta \theta'} \rVert_n = \psi^*_{\theta \theta'}$ for all $n \ge 1$. In addition, all the inequalities used to construct the bounds above are equalities for this case, so we conclude that
\begin{align}
    {\cal T}(h^*,c) + {\cal H}(h^*) = Q(\psi^*).
\end{align}
Therefore, the graphon $h^*$ solves the variational problem in Theorem \ref{thm:colored_graphon_partition}. 

To show that only graphons that are constant almost everywhere solve the problem, assume for the sake of contradiction that there is a solution $\hat{h}$ that is not constant almost everywhere. Then the inequalities above are strict, and we would obtain a strict improvement by ``flattening'' the solution in a way that preserves the norms of the subgraphons $\lVert \hat{h}_{\theta \theta'} \rVert_{e^*}$.

The convergence to the directed stochastic block model when the optimizer is unique is a direct consequence of Theorem \ref{thm:colored_graphon_convergence}.


\pagebreak

\section{Graph Limits} \label{sec:app_graph_limits}
In this section, I present the relevant definitions and results for the theory of large dense graphs. The results in Sections \ref{subsec:graphons} and \ref{subsec:results_graph_limits} are presented without proof. These are taken from Appendix D in \citet{mele_structural_2017}, which is an excellent introduction to the topic, and builds on the results in \citet{chatterjee_large_2011} and \citet{chatterjee_estimating_2013}. The reader is invited to read this, and the references therein, for a detailed discussion on the topic\footnote{Also see \citet{lovasz_large_2012} for a broader introduction to graph limits.}. Section \ref{subsec:colored_graphons} extends these results to colored graphs, allowing us to characterize the limiting behavior of the model with heterogeneous agents. In order to make a clearer connection to the literature on graph limits, the notation in this section will differ from the rest of the paper. 
  
\subsection{Graphons} \label{subsec:graphons}
This section presents the relevant definitions for the theory of graph limits, following the outline presented in \citet{chatterjee_estimating_2013} adapted to directed graphs. For this appendix, a directed graph is an ordered pair $G = (V,A)$, where $V$ is a set of vertices and $A$ is a set of arcs (directed edges). Consider a sequence $G_N$ of simple directed graphs whose number of nodes tends to infinity. For every fixed simple graph $H$, let $|\textrm{hom}(H,G)|$ denote the number of homomorphisms of $H$ into $G$, which is the number of edge-preserving maps from $V(H)$ into $V(G)$. That is, a map $\varphi:V(H) \to V(G)$ is a homomorphism if $(i,j)\in A(H) \implies (\varphi(i),\varphi(j)) \in A(G)$. Normalizing by the possible number of maps from $V(H)$ to $V(G)$ yields the \textit{homomorphism density}
\begin{align}
    t(H,G) \coloneqq \frac{|\textrm{hom}(H,G)|}{|V(G)|^{|V(H)|}},
\end{align}
which corresponds to the probability of a uniformly random mapping $V(H) \to V(G)$ being a homomorphism.

The importance of analyzing homomorphism densities is twofold. First, they provide a way to ``probe'' large graphs in order to understand their properties. Second, they're relevant in our context since motif utilities can be approximated as homomorphism densities for large graphs, as is discussed below.

Motif utilities do not depend explicitly on the number of homomorphisms of a motif into the network, but rather the number of times it appears as a subgraph. This can be captured using the concept of \textit{subgraph densities}. Let $|\textrm{sub}(H,G)|$ be the number of \textit{injective} maps from $V(H)$ to $V(G)$ that are homomorphisms. Clearly $|\textrm{sub}(H,G)| \le |\textrm{hom}(H,G)|$. The normalized motif densities in our model are given by subgraph densities, defined by
\begin{align}
    s(H,G) \coloneqq \frac{|\textrm{sub}(H,G)|}{|V(G)|^{|V(H)|}}.
\end{align}
It can be shown that
\begin{align}
    t(H,G) - \frac{1}{|V(G)|} {|V(H)| \choose 2} \le s(H,G) \le t(H,G).
\end{align}
Therefore, for a given $H$, characterizing homomorphism densities for large graphs is equivalent to characterizing subgraph densities.

A possible notion of convergence for the sequence $\{G_N\}$ is that the densities $t(H,G_N)$ converge to some value for every finite graph $H$. The work of Lov\'asz and coauthors (see \citet{lovasz_large_2012} for an overview) established the existence of an object that characterizes this convergence for undirected graphs. That is, there exists an object from which the limiting homomorphism densities $t(H,\cdot)$ can be obtained. An analogous result for directed graphs was established in \citet{boeckner_directed_2013}. This ``limiting object'' is a function $h \in {\cal W}$, where ${\cal W}$ is the set of measurable functions $[0,1]^2 \to [0,1]$. These objects are called ``graphons''. Conversely, every function in ${\cal W}$ arises as the limit of an appropriate sequence of directed graphs. If $H$ is a simple directed graph with $V(H) = \{1,\ldots,k\}$, the homomorphism density of $H$ into $h$ is defined as
\begin{align}
    t(H,h) \coloneqq \int_{[0,1]^k} \left( \prod_{(i,j) \in A(H)} h(x_i, x_j) \right) \, dx_1 \ldots dx_k.
\end{align}
A sequence of graphs $\{G_N\}$ is said to converge to $h$ if 
\begin{align}
    \lim_{N \to \infty} t(H,G_N) = t(H,h)
\end{align}
for all finite simple directed graphs $H$.

The intuition is that as $N \to \infty$, the interval $[0,1]$ represents a ``continuum'' of vertices and $h(x,y)$ is the probability that there is an arc going from $x$ to $y$.\footnote{This intuition is more explicit in the case of $W$ random graphs.} For the case of directed Erd\H{o}s--R\'enyi graphs $G(N,p)$ with fixed $p$, the limit graph is represented by the graphon that is equal to $p$ for all $(x,y) \in [0,1]^2$. For a fixed $p \in (0,1)$, this corresponds to a dense random graph model.

Finite simple graphs have a canonical representation as graphons. For a given graph $G$ over $N$ nodes, its associated graphon is given by
\begin{align}
    h^G(x,y) \coloneqq \mathbbm{1}\{ (\lceil Nx \rceil, \lceil Ny \rceil) \in A(G) \}.
\end{align}

The notion of convergence in ${\cal W}$ in terms of homomorphism densities can be metrized using the \textit{cut distance}. For two graphons $h_1$ and $h_2$, it is defined as
\begin{align}
    d_\square(h_1,h_2) \coloneqq \sup_{S, T \subseteq [0,1]} \left| \int_{S \times T} [h_1(x,y) - h_2(x,y)] \, dx \, dy \right|.
\end{align}
For our purposes, it will be useful to work in a different space that has some useful topological properties. Let $\Sigma$ be the set of measure-preserving bijections $\sigma: [0,1] \to [0,1]$. Define an equivalence relation on ${\cal W}$ by setting $h_1 \sim h_2$ if $h_1(x,y) = h_2^\sigma(x,y) \coloneqq h_2(\sigma x, \sigma y)$ for some $\sigma \in \Sigma$. Let $\tilde{h}$ be the closure of the orbit $\{h^\sigma\}$ in $({\cal W},d_\square)$. Let $\tilde{{\cal W}} \coloneqq {\cal W}/\sim$ and let $\tau$ be the map $\tau h \mapsto \tilde{h}$. The distance $\delta_\square$ on the space $\tilde{{\cal W}}$ is defined as
\begin{align}
    \delta_\square(\tilde{h}_1, \tilde{h}_2) = \inf_{\sigma} d_\square(h_1, h_2^\sigma),
\end{align}
such that $(\tilde{{\cal W}}, \delta_\square)$ is a metric space. For any finite directed graph $G$, let $G$ be its associated graphon and let $\tilde{G} \coloneqq \tilde{h}^G \in \tilde{{\cal W}}$ be its corresponding orbit.

An important aspect of homomorphism densities is that they are continuous functions in this space. This is useful in the analysis of the limiting behavior of ERGMs, as discussed below.

\subsection{Results on graph limits} \label{subsec:results_graph_limits}
A central tool used in characterizing the limiting behavior of ERGMs is the theory of large deviations for random graph models. Fix some $p \in (0,1)$. Let $\tilde{\mathbb{P}}_{N,p}$ be the measure induced on $\tilde{{\cal W}}$ by the directed Erd\H{o}s--R\'enyi model with parameter $p$. Additionally, define the function $I_p:[0,1] \to \mathbb{R}$ as
\begin{align}
    I_p(u) \coloneqq u \log\left( \frac{u}{p} \right) + (1-u) \log\left( \frac{1-u}{1-p} \right).
\end{align}
The domain of this function can be extended to $\tilde{{\cal W}}$ by defining
\begin{align} \label{eq:er_rate_func}
    {\cal I}_p(\tilde{h}) \coloneqq \int_{[0,1]^2} I_p(h(x,y)) \, dx \, dy
\end{align}
for any $h \in \tau^{-1}(\tilde{h})$. Building off the results in \citet{chatterjee_large_2011}, we have the following result for directed graphs.
\begin{theorem}[Theorem 8 in \citet{mele_structural_2017}]\label{thm:er_ldp}
    For each fixed $p \in (0,1)$, the sequence $\tilde{\mathbb{P}}_{N,p}$ satisfies a large deviation principle on the space $(\tilde{{\cal W}}, \delta_\square)$ at speed $N^2$ with rate function ${\cal I}_p$. Explicitly, for any closed set $\tilde{F} \subseteq \tilde{{\cal W}}$,
    \begin{align}
        \limsup_{N \to \infty} \frac{1}{N^2} \log(\tilde{\mathbb{P}}_{N,p}(\tilde{F})) \le - \inf_{\tilde{h} \in \tilde{F}} {\cal I}_p(\tilde{h})
    \end{align}
    and for any open set $\tilde{U} \subseteq \tilde{{\cal W}}$,
    \begin{align}
        \liminf_{N \to \infty} \frac{1}{N^2} \log(\tilde{\mathbb{P}}_{N,p}(\tilde{U})) \ge - \inf_{\tilde{h} \in \tilde{U}} {\cal I}_p(\tilde{h}).
    \end{align}
\end{theorem}

Now suppose we have a sequence of measures $\pi_N$ on ${\cal G}_N$ given by
\begin{align}
    \pi_N(G) = \exp\{ N^2 ({\cal T}(\tilde{G}) - \psi_N) \},
\end{align}
where $\psi_N$ is a normalization constant given by
\begin{align}
    \psi_N \coloneqq \frac{1}{N^2} \log\left( \sum_{G \in {\cal G}_N} e^{N^2 {\cal T}(\tilde{G})} \right).
\end{align}
In the case of \citet{mele_structural_2017} and our motif utilities, the function ${\cal T}$ is a linear combination of homomorphism densities. Define the \textit{entropy functional} ${\cal H}:\tilde{{\cal W}} \to \mathbb{R}$
\begin{align} \label{eq:graphon_entropy}
    {\cal H}(\tilde{h}) = \int_{[0,1]^2} H(h(x,y)) \, dx \, dy
\end{align}
for some $h \in \tau^{-1}(\tilde{h})$ and $H$ is given by
\begin{align}
    H(u) \coloneqq - u \log(u) - (1-u) \log(1-u).
\end{align}
Adapting the results in \citet{chatterjee_estimating_2013}, \citet{mele_structural_2017} finds the following result on the normalization constant of these models.

\begin{theorem}[Theorem 10 in \citet{mele_structural_2017}]\label{thm:partition_func_convergence}
    If ${\cal T}:\Tilde{{\cal W}} \to \mathbb{R}$ is a bounded continuous function, then
    \begin{align} \label{eq:partition_func_convergence}
        \psi \coloneqq \lim_{N \to \infty} \psi_N = \sup_{\Tilde{h} \in \Tilde{{\cal W}}} \left[ {\cal T}(\Tilde{h}) + {\cal H}(\Tilde{h}) \right].
    \end{align}
\end{theorem}

Using the properties of the space $\tilde{{\cal W}}$, an even stronger result on convergence can be obtained. This allows us to characterize the limiting behavior of the statistical properties of ERGMs.

\begin{theorem}[Theorem 18 in \citet{mele_structural_2017}]\label{thm:graphon_convergence}
    Let $\Tilde{M}^*$ be the set of maximizers of the variational problem \eqref{eq:partition_func_convergence}. Let $G_N$ be a graph on $N$ vertices drawn from the model implied by function ${\cal T}$. Then, for any $\eta > 0$, there exist $C, \kappa > 0$ such that, for any $N$,
    \begin{align*}
        \mathbb{P}\{ \delta_\square(\Tilde{G}_N, \Tilde{M}^*) > \eta \} \le C e^{-N^2 \kappa},
    \end{align*}
    where $\mathbb{P}$ denotes the probability measure implied by the model.
\end{theorem}

Building on these results, we can characterize the limiting behavior of models with heterogeneity in the nodes.

\subsection{Colored graphons} \label{subsec:colored_graphons}
In order to account for node heterogeneity, I use the framework of \textit{colored graphs}, particularly as laid out in \citet{diao_model-free_2016}. Let $\Theta$ be a set of colors (or types in our case). A directed colored graph $Q$ is a tuple $(V,A,C)$, where $C:V \to \Theta$ represents the coloring of the vertices of $Q$. Finally, let ${\cal Q}_N^\Theta$ be the set of colored graphs with color set $\Theta$.

Similar to the case of uncolored graphs, we can consider a way to ``probe'' colored graphs using the homomorphism densities of other colored graphs into them. Let $R$ and $Q$ be directed colored graphs. We define $|\textrm{hom}_\Theta(R,Q)|$ to be the number of homomorphisms of $(V(R),A(R))$ into $(V(Q),A(Q))$ that preserve the coloring of the vertices. The \textit{colored homomorphism density} of $R$ into $Q$, then, is
\begin{align}
    t_\Theta(R,Q) \coloneqq \frac{|\textrm{hom}_\Theta(R,Q)|}{|V(Q)|^{|V(R)|}}.
\end{align}
As in the uncolored case, we can characterize the convergence of a sequence of colored graphs by studying their colored homomorphism densities. In order to do this, we need to define appropriate limiting objects.

A colored directed graphon is a pair $q \coloneqq (h[q], c[q])$, where $h[q] \in {\cal W}$ is a directed graphon and $c[q]$ is a measurable function $c[q]: [0,1] \to \Theta$, which I denote the \textit{coloring}. The space of colored directed graphons is denoted with ${\cal W}_\Theta$. The colored homomorphism density of $R$ into $q$, assuming $V(R) = \{1,\ldots,k\}$, is defined as
\begin{align}
    t_\Theta(R,q) \coloneqq \int_{[0,1]^{k}} \left( \prod_{(i,j) \in A(R)} h[q](x_i,x_j) \right) \left( \prod_{i=1}^k \mathbbm{1}\{ c[q](x_i) = C(R)_i \} \right) \, dx_1 \ldots dx_{k}.
\end{align}
As in the uncolored case, there is a canonical representation of colored graphs as colored graphons. For a given colored graph $Q$, its associated graphon $q^Q = (h[q^Q], c[q^Q])$ is given by
\begin{align}
    h[q^Q] = h^{(V(Q),A(Q))}, \quad c[q^Q](x) = C(Q)_{\lceil |V(Q)| x \rceil}.
\end{align}
Note that for the uncolored graph $G=(V(Q),A(Q))$, we have $h[q^Q] = h^G$.

We can also construct a topology in the space of colored graphons. In order to do this, we need a new notion of distance. The \textit{colored cut distance} between two graphons is given by
\begin{align}
    d_\square^\Theta(q_1, q_2) \coloneqq d_\square(h[q_1], h[q_2]) + d_\Theta(c[q_1],c[q_2]),
\end{align}
where
\begin{align}
    d_\Theta(c, c') \coloneqq \frac{1}{2} \sum_{\theta \in \Theta} \int_{[0,1]} \mathbbm{1}\{ x \in c^{-1}(\theta) \Delta c'^{-1}(\theta) \} \, dx
\end{align}
is a distance on the space of colorings.\footnote{Equivalently, it is the Lebesgue measure of the set where $c$ and $c'$ do not agree.} The equivalence class in this space is defined similarly to the uncolored case, but we have to take into account that the bijections $\sigma$ also act on the colorings. Define the equivalence relation $\sim_\Theta$ by setting $q_1 \sim_\Theta q_2$ if $h[q_1](x,y) = h[q_2^\sigma](x,y) = h[q_2](\sigma x, \sigma y)$ and $c[q_1](x) = c[q_2^\sigma](x) = c[q_2](\sigma x)$ for some $\sigma \in \Sigma$. Define $\tilde{q}$ to be the equivalence class of $q$, and let $\tilde{{\cal W}}_\Theta \coloneqq {\cal W}_\Theta/\sim_\Theta$. Additionally, let $\tau_\Theta$ be the map such that $\tau_\Theta(q) = \tilde{q}$. For any finite colored graph $Q$, let $\tilde{Q} \coloneqq \tilde{q}^Q \in \tilde{{\cal W}}_\Theta$ be its corresponding orbit. The natural distance on this space is given by
\begin{align}
    \delta_\square^\Theta(\tilde{q}_1, \tilde{q}_2) \coloneqq \inf_\sigma d_\square^\Theta(q_1, q_2^\sigma).
\end{align}
It is also convenient to define an analogous equivalence on the space of colorings. Let ${\cal C}$ be the set of measurable functions $c: [0,1] \to \Theta$ and define $\tilde{{\cal C}}$ as the corresponding space of orbits. Define the coloring distance in this space as
\begin{align}
    \delta_\Theta(\tilde{c}_1,\tilde{c}_2) = \inf_\sigma d_\Theta(c_1,c_2^\sigma).
\end{align}
Note that every coloring orbit $\tilde{c}$ is uniquely identified (up to null sets) with a color vector ${\bf w}(\tilde{c}) \coloneqq (\lambda(c^{-1}(\theta))_{\theta \in \Theta} \in \Delta(\Theta)$, where $\lambda$ is the Lebesgue measure, for any representative $c \in \tilde{c}$. Therefore, we can also uniquely define a coloring orbit $\tilde{c}_{{\bf w}}$ for a given color vector ${\bf w}$. This equivalence with color vectors allows us to prove the following property of the space of colorings:
\begin{lemma}
    The space $(\tilde{{\cal C}}_\Theta, \delta_\Theta)$ is compact.
\end{lemma}
\begin{proof}
    To prove this, I show that $(\tilde{{\cal C}}_\Theta, \delta_\Theta)$ is isometric, up to a factor,9 with $\Delta(\Theta)$ equipped with the $L^1$ norm. This is a closed and bounded subset of $\mathbb{R}^{|\Theta|}$, hence compact by Heine-Borel.

    Consider two colorings $\tilde{c}$ and $\tilde{c}'$. The sets that are assigned to $\theta$ cannot overlap over more than $\min\{w_\theta(\tilde{c}), w_\theta(\tilde{c}')\}$. This implies that $\delta_\Theta(\tilde{c}, \tilde{c}') \ge \lVert {\bf w}(\tilde{c}) - {\bf w}(\tilde{c}') \rVert_1$. Additionally, note that this bound can be achieved by ordering both colorings by intervals of length $\min\{w_\theta(\tilde{c}), w_\theta(\tilde{c}')\}$, and having the remaining intervals be mismatched. Hence,
    \begin{align}
        \delta_\Theta(\tilde{c}, \tilde{c}') = \frac{1}{2} \lVert {\bf w}(\tilde{c}) - {\bf w}(\tilde{c}') \rVert_1.
    \end{align}
    This proves the isometry (with a factor of $1/2$), which implies compactness of $(\tilde{{\cal C}}_\Theta, \delta_\Theta)$.
\end{proof}

The space $(\tilde{{\cal W}}_\Theta, \delta^\Theta_\square)$ has similar topological properties to $(\tilde{{\cal W}}, \delta_\square)$. There is one particular result that is of interest to us, since it allows us to characterize the limiting behavior of the measures induced by random graph models.
\begin{theorem}
    The space $(\tilde{{\cal W}}_\Theta, \delta_\square^\Theta)$ is compact.
\end{theorem}

\begin{proof}
    This proof follows similarly to the proof of Theorem 3.7 in \citet{diao_model-free_2016} for the undirected case. The strategy is to prove that any sequence of colored graphons $q_1, q_2, \ldots$ has a convergent subsequence. Since $(\tilde{{\cal W}}_\Theta, \delta_\square^\Theta)$ is a metric space, this would imply it is compact.

    First, note that there is a sequence of measure-preserving bijections $\sigma_n:[0,1] \to [0,1]$ such that the colored graphons $q_n^{\sigma_n}$ have colorings $c[q_n^{\sigma_n}]$ that are intervals that follow some fixed order of $\Theta$. Therefore, without loss of generality, we can consider that the colored graphons $q_n$ already have colorings arranged by intervals. Since the vectors ${\bf w}(\tilde{c}[q_n])$ lie on the simplex (which is compact), there exists a subsequence $n_k$ that converges to some color vector ${\bf w}_0$. This implies that the limit
    \begin{align}
        \lim_{k \to \infty} c[q_{n_k}](x) \eqqcolon c_0(x)
    \end{align}
    exists almost everywhere, and $\lambda(c[q_{n_k}]^{-1}(\theta) \Delta c_0^{-1}(\theta)) \to 0$ for all $\theta \in \Theta$. 
    
    The proof now follows the proof of Lemma 5 in \citet{mele_structural_2017}, which adapts the proof of Theorem 9.23 in \citet{lovasz_large_2012} to directed graphons. Since we know we can take a subsequence whose coloring converges from any sequence of colored graphons, suppose the colorings $c[q_n]$ converge. Now apply the argument from the proof of Lemma 5 in \citet{mele_structural_2017} to the graphons $h[q_n]$. Note that the partitions ${\cal P}_{n,1}$ can be chosen to respect the partition defined by $c[q_n]^{-1}(\theta)$ for $\theta \in \Theta$. Since for $k > 1$ the partitions ${\cal P}_{n,k}$ are refinements of this one, they also respect the coloring partition. The procedure in that proof then gives us a subsequence $q_{n_l}$ such that $d_\square(h[q_{n_l}],h_0)  \to 0$ for some graphon $h_0$.\footnote{The proof of Lemma 5 in \citet{mele_structural_2017} technically yields a sequence that converges under $\delta_\square$. However, it invokes Lemmas 3.1.20 and 3.1.21 in \citet{boeckner_directed_2013}, which state properties of convergence under $d_\square$. Using these results with $d_\square$ and carrying out the argument in the proof of Lemma 5 in \citet{mele_structural_2017} yields $d_\square(h[q_{n_l}],h_0) \to 0$.} Now, define the graphon $q_0 \coloneqq (h_0, c_0)$. We have that
    \begin{align}
        \delta_\square^\Theta(q_{n}, q_0) \le d_\square(h[q_n],h_0) + d_\Theta(c[q_n],c_0).
    \end{align}
    Under the subsequence $(n_l)_l$, both terms on the right vanish, meaning $q_{n_l} \to q_0$ in the metric space $(\tilde{{\cal W}}_\Theta, \delta_\square^\Theta)$.
\end{proof}

Having this topological property, we can now state the result that allows for the generalization of the graph convergence results to the colored case. Consider a sequence of empirical fractions of colors ${\bf w}^N$ that converges to some ${\bf w}$. Consider the measures $\Tilde{\mathbb{P}}^\Theta_{N,p}$ on $\tilde{{\cal W}}_\Theta$ generated by the following procedure:
\begin{enumerate}
    \item Take the set of $N$ nodes $\{1,\ldots,N\}$ and deterministically fix a coloring such that the coloring fractions are ${\bf w}_N$.
    \item Draw a directed Erd\H{o}s--R\'enyi random graph on the nodes with parameter $p$.
\end{enumerate}
Note that the specific coloring of the nodes does not matter, since one could relabel the nodes and run the same procedure without affecting the measures. The next theorem characterizes the limiting behavior of these measures, which later allows us to characterize the asymptotic behavior of the Gibbs measures for Theorem \ref{thm:mult_types_partition}.

\begin{theorem} \label{thm:colored_graphon_ldp}
    For each fixed $p \in (0,1)$, the sequence $\Tilde{\mathbb{P}}^\Theta_{N,p}$ satisfies a large deviation principle in the space ($\Tilde{{\cal W}}_\Theta$, $\delta_{\square}^\Theta$) at speed $N^2$ with rate function
    \begin{align}
        {\cal J}^\Theta_{{\bf w},p}(\tilde{q}) = 
        \begin{cases}
            {\cal I}_p(\tilde{h}[\tilde{q}]) & \textrm{if } \tilde{c}[\tilde{q}] = \tilde{c}_{{\bf w}}, \\
            \infty & \textrm{otherwise},
        \end{cases}
    \end{align}
    where ${\cal I}_p$ is defined in Equation \eqref{eq:er_rate_func}. Explicitly, for any closed set $\tilde{F} \subseteq \tilde{{\cal W}}_\Theta$,
    \begin{align}
        \limsup_{N \to \infty} \frac{1}{N^2} \log(\tilde{\mathbb{P}}^\Theta_{N,p}(\tilde{F})) \le - \inf_{\tilde{q} \in \tilde{F}} {\cal J}^\Theta_{{\bf w},p}(\tilde{q})
    \end{align}
    and for any open set $\tilde{U} \subseteq \tilde{{\cal W}}_\Theta$,
    \begin{align}
        \liminf_{N \to \infty} \frac{1}{N^2} \log(\tilde{\mathbb{P}}^\Theta_{N,p}(\tilde{U})) \ge - \inf_{\tilde{q} \in \tilde{U}} {\cal J}^\Theta_{{\bf w},p}(\tilde{q}).
    \end{align}
\end{theorem}

\begin{proof}
    In order to prove the LDP result on the space $\tilde{{\cal W}}_\Theta$, I will use the fact that an LDP is satisfied for the uncolored graphons. Given the structure of the process that generates the random colored graphons, enough structure is preserved (as needed in Lemma \ref{lem:inverse_contraction_ldp}) to conclude that these also satisfy an LDP.

    First, I will prove that an LDP is satisfied in the product space ${\cal W} \times {\cal C}_\Theta$, with the product topology. Define the function $\Psi: \tilde{{\cal W}}_\Theta \to \tilde{{\cal W}} \times \tilde{C}_\Theta$ by
    \begin{align}
        \Psi(\tilde{q}) \coloneqq (\tilde{h}[\tilde{q}], \tilde{c}[\tilde{q}]).
    \end{align}
    This function ``separates'' the orbit of a colored graphon into its uncolored graphon and its coloring. Clearly, many colored graphons can have the same components $(\tilde{h}, \tilde{c})$, so the function $\Psi$ is not injective.
    
    Note that the product topology can be induced by the metric
    \begin{align}
        \delta_{\otimes}^{\Theta}((\tilde{h}_1, \tilde{c}_1), (\tilde{h}_2, \tilde{c}_2)) \coloneqq \delta_\square(\tilde{h}_1, \tilde{h}_2) + \delta_\Theta(\tilde{c}_1, \tilde{c}_2).
    \end{align}
    It is easy to see that $\Psi$ is continuous using this metric. Indeed, by definition of $\delta_\square^\Theta$,
    \begin{align}
        \delta_\square^\Theta(\tilde{q}_1, \tilde{q}_2) &= \inf_\sigma \left[ d_\square(h[q_1], h[q_2^\sigma]) + d_\Theta(c[q_1],c[q_2^\sigma]) \right] \nonumber \\
        &\ge  \delta_{\otimes}^{\Theta}((\tilde{h}[\tilde{q}_1], \tilde{c}[\tilde{q}_1]), (\tilde{h}[\tilde{q}_2], \tilde{c}[\tilde{q}_2])) \nonumber \\
        &= \delta_{\otimes}^{\Theta}(\Psi(\tilde{q}_1), \Psi(\tilde{q}_2)),
    \end{align}
    the mapping is Lipschitz and, therefore, continuous. It is also surjective, since for any orbits $\tilde{h}$ and $\tilde{c}$, a colored graphon can be constructed from elements of these orbits, and has a corresponding orbit in $\tilde{{\cal W}}_\Theta$. Also note that since $\tilde{{\cal W}}$ and $\tilde{{\cal C}}_\Theta$ are compact metric spaces, so is $\tilde{{\cal W}} \times \tilde{{\cal C}}_\Theta$.

    Let us show that the pushforward measures $\tilde{\mathbb{P}}^\Theta_{N,p} \circ \Psi^{-1}$ satisfy an LDP. Note that these are simply the product measure of the Erd\H{o}s--R\'enyi measure on $\tilde{{\cal W}}$ and the Dirac measure centered on the correct empirical coloring on $\tilde{{\cal C}}_\Theta$:
    \begin{align}
        \tilde{\mathbb{P}}^\Theta_{N,p} \circ \Psi^{-1} = \Tilde{\mathbb{P}}_{N,p} \otimes \delta_{\tilde{c}_{{\bf w}^N}},
    \end{align}
    since the uncolored graphons are generated independent of the colorings. The complex structure that comes from coupling the graphon to the coloring is lost once they are decoupled using $\Psi$. Because the empirical distributions converge to ${\bf w}$, it follows that the measures $\delta_{\tilde{c}_{{\bf w}^N}}$ satisfy an LDP at \textit{any} speed with rate function
    \begin{align}
        {\cal I}^\Theta_{{\bf w}}(\tilde{c}) =
        \begin{cases}
            0 & \textrm{if } \tilde{c} = \tilde{c}_{{\bf w}}, \\
            \infty & \textrm{otherwise}.
        \end{cases}
    \end{align}
    In particular, this implies that these measures satisfy a large deviations principle at speed $N^2$. Additionally, note that ${\cal I}^\Theta_{{\bf w}}$ and ${\cal I}_p$ are good rate functions, since $\tilde{{\cal C}}_\Theta$ and $\tilde{{\cal W}}$ are compact metric spaces. Using Lemma \ref{lem:prod_measures_ldp} and Theorem \ref{thm:er_ldp}, we have that the measures $\tilde{\mathbb{P}}^\Theta_{N,p} \circ \Psi^{-1}$ satisfy a large deviations principle on the product space $\tilde{{\cal W}} \times \tilde{{\cal C}}_\Theta$ at speed $N^2$ with rate function
    \begin{align}
        {\cal I}^{\Theta,\textrm{prod}}_{{\bf w}, p}(\tilde{h}, \tilde{c}) =
        \begin{cases}
            {\cal I}_p(\tilde{h}) & \textrm{if } \tilde{c} = \tilde{c}_{{\bf w}}, \\
            \infty & \textrm{otherwise.}
        \end{cases}
    \end{align}

    Having shown that the pushforward measures satisfy an LDP in $\tilde{{\cal W}} \times \tilde{{\cal C}}_\Theta$, we need to show that the conditions in Lemma \ref{lem:inverse_contraction_ldp} are satisfied to obtain an LDP in $\tilde{{\cal W}}_\Theta$. 
    
    First, note that the support of $\Tilde{\mathbb{P}}_{N,p} \otimes \delta_{\tilde{c}_{{\bf w}^N}}$ is composed of the orbits $(\tilde{G},\tilde{c}_{{\bf w}^N})$, where $G \in {\cal G}_N$ are directed graphs over $N$ vertices, which is a discrete set. The fibers $\Psi^{-1}((\tilde{G},\tilde{c}_{{\bf w}_N}))$ are not discrete, but they contain the set of colored directed graphs with color fractions ${\bf w}_N$ and graphon orbit $\tilde{G}$. We will choose this to be our set $\Upsilon_N$ (in the notation of Lemma \ref{lem:inverse_contraction_ldp}):
    \begin{align}
        \Upsilon_N((\tilde{G},\tilde{c}_{{\bf w}^N})) \coloneqq \Psi^{-1}((\tilde{G},\tilde{c}_{{\bf w}^N})) \cap \tilde{{\cal Q}}_N^\Theta,
    \end{align}
    where $\tilde{{\cal Q}}_N^\Theta \subset \tilde{{\cal W}}_\Theta$ is the set of colored graphon orbits generated from directed colored graphons over $N$ nodes with colorings over $\Theta$.

    I now relate the probability associated to a colored graphon to the probability of its corresponding uncolored graphon. Consider colored graphs over $N$, let $\tau$ be a permutation of $\{1,\ldots,N\}$ and let $C_0^N$ be the (fixed) coloring used to generate $\tilde{\mathbb{P}}_{N,p}^\Theta$. For a finite graph $G$ with $N$ nodes, the probability associated to $\tilde{G}$ is
    \begin{align}
        \tilde{\mathbb{P}}_{N,p}(\{\tilde{G}\}) = p^{|A(G)|} (1-p)^{N(N-1)-|A(G)|} \left(\frac{N!}{\sum_{\tau} \mathbbm{1}\{\tau A(G) = A(G)\}} \right).
    \end{align}    
    The first factor is the probability of drawing a specific representative directed graph. The second counts the number of directed graphs with orbit $\tilde{G}$: it counts all re-labelings of the nodes and takes care of degeneracy by counting the automorphisms of $G$. Similarly, for a colored graph $Q$ over $N$ nodes (with the correct coloring), its associated probability is
    \begin{align}
        \tilde{\mathbb{P}}^\Theta_{N,p}(\{\tilde{Q}\}) = p^{|A(Q)|} (1-p)^{N(N-1)-|A(Q)|} \left( \frac{\prod_{\theta \in \Theta} (N w_\theta^N)!}{\sum_{\tau} \mathbbm{1}\{\tau A(Q) = A(Q), \tau C(Q) = C_0^N\}} \right).
    \end{align}
    Now the re-labelings happen within each color, and the degeneracy must also take into account permutations that respect the coloring. Since the matching of the coloring is an additional restriction, we can see that for the restricted graph $G = (V(Q),A(Q))$:
    \begin{align}
        \tilde{\mathbb{P}}^\Theta_{N,p}(\{\tilde{Q}\}) \ge \frac{\left( \prod_{\theta \in \Theta} (N w_\theta^N)! \right)}{N!} \tilde{\mathbb{P}}_{N,p}(\{\tilde{G}\}).
    \end{align}
    Define the weights $\alpha_N(\tilde{Q},(\tilde{G},\tilde{c}_{{\bf w}^N}))$ by
    \begin{align}
        \alpha_N(\tilde{Q},(\tilde{G},\tilde{c}_{{\bf w}^N}) = \tilde{\mathbb{P}}^\Theta_{N,p}(\{\tilde{Q}\} | \{\tilde{G}\}) \quad \textrm{for } \tilde{Q} \in \Upsilon_N((\tilde{G},\tilde{c}_{{\bf w}^N})).
    \end{align}
    From the argument above, we have the following bound on these weights:
    \begin{align}
        \alpha_N(\tilde{Q},(\tilde{G},\tilde{c}_{{\bf w}^N}) \ge \frac{\left( \prod_{\theta \in \Theta} (N w_\theta^N)! \right)}{N!}.
    \end{align}
    It is useful to note the Stirling bounds on the factorial:
    \begin{align}
        \sqrt{2 \pi n}\, n^{\,n+\tfrac{1}{2}} e^{-n} \le n! \le e\, n^{\,n+\tfrac{1}{2}} e^{-n}.
    \end{align}
    With this, we conclude that
    \begin{align}
        \liminf_{N \to \infty} \inf_{G \in {\cal G}_N} \min_{\tilde{Q} \in \Upsilon_N((\tilde{G},\tilde{c}_{{\bf w}^N}))} \frac{1}{N^2} \log(\alpha_N(\tilde{Q},(\tilde{G},\tilde{c}_{{\bf w}^N}))) = 0.
    \end{align}
    Hence, the subexponential weight floor condition in Lemma \ref{lem:inverse_contraction_ldp} is satisfied.

    The uniform fiber approximation condition is satisfied by the result in Lemma \ref{lem:fiber_density}, with
    \begin{align}
        \epsilon_N = \frac{\kappa_1}{\sqrt{\log(N)}} + 2|\Theta| N^{-\kappa_2}
    \end{align}
    for some constants $\kappa_1$ and $\kappa_2$. Therefore, applying Lemma \ref{lem:inverse_contraction_ldp} we have that $\Tilde{\mathbb{P}}^\Theta_{N,p}$ satisfies a large deviations principle on $\tilde{{\cal W}}_\Theta$ at speed $N^2$ with rate function ${\cal J}^\Theta_{{\bf w}, p}$. This completes the proof.
\end{proof}

Using the previous result, we can state the convergence theorems that yield Theorem \ref{thm:mult_types_partition} in the paper.

\begin{theorem} \label{thm:colored_graphon_partition}
    Consider a sequence of colorings $C_0^1, C_0^2, \ldots$ such that their empirical color fractions ${\bf w}^N$ converge to some ${\bf w}$. Additionally, consider an Exponential Random Graph model where the measure is given by
    \begin{align}
        \pi_N(G) = \exp\left\{ N^2 [{\cal U}(\Tilde{Q}(G)) - \psi_N] \right\},
    \end{align}
    where $Q(G) = (G, C_0^N)$ is the colored directed graph induced by $G$ and the coloring $C_0^N$. If ${\cal U}$ is a continuous bounded function on $\tilde{{\cal W}}_\Theta$, then the normalization constant $\psi_N$ given by
    \begin{align}
        \psi_N = \frac{1}{N^2} \log\left( \sum_{G \in {\cal G}_N} e^{N^2 {\cal U}(\tilde{Q}(G))} \right)
    \end{align}
    converges to
    \begin{align} \label{eq:colored_graphon_partition}
        \psi \coloneqq \sup_{\substack{\tilde{q} \in \Tilde{{\cal W}}_\Theta \\ \tilde{c}[\tilde{q}] = \tilde{c}_{{\bf w}}}} ({\cal U}(\Tilde{q}) + {\cal H}(\tilde{h}[\tilde{q}])).
    \end{align}
\end{theorem}

\begin{proof}
    This proof follows the same procedure as the proof of Theorem 3.1 in \citet{chatterjee_estimating_2013} for undirected graphons and Theorem 10 of \citet{mele_structural_2017} for directed graphons. Let $\tilde{B} \subseteq \tilde{{\cal W}}_\Theta$ be a Borel set. For each $N$, let $B_N \subseteq {\cal W}_\Theta$ be the (finite) set
    \begin{align}
        B_N \coloneqq \{ q \in {\cal W}_\Theta: \tilde{q} \in \tilde{B} \textrm{ and } q = q^Q \textrm{ for some } Q \textrm{ with } C(Q) = C_0^N \textrm{ and } (V(Q),A(Q)) \in {\cal G}_N \}.
    \end{align}
    That is, for all $q \in B_N$, there is a colored directed graph $Q$ that has the correct coloring $C_0^N$ and corresponds to a directed graph over $N$ nodes. Recall that $\tilde{\mathbb{P}}^\Theta_{N,p}$ corresponds to the measure induced on $\tilde{{\cal W}}_\Theta$ by fixing the correct coloring $C_0^N$ and sampling the edges using a directed Erd\H{o}s--R\'enyi model. Therefore, we have that
    \begin{align}
        |B_N| = 2^{N(N-1)} \tilde{\mathbb{P}}^\Theta_{N,1/2}(\tilde{B}).
    \end{align}
    
    Using the result from Theorem \ref{thm:colored_graphon_ldp}, we have that for all closed $\tilde{F} \subseteq \tilde{{\cal W}}_\Theta$,
    \begin{align}
        \limsup_{N \to \infty} \frac{1}{N^2} \log(\tilde{\mathbb{P}}^\Theta_{N,p}(\tilde{F})) &= \limsup_{N \to \infty} \frac{1}{N^2} \left[\log(|F_N|) - N(N-1) \log(2) \right]\nonumber \\
        &= \limsup_{N \to \infty} \frac{1}{N^2} \log(|F_N|) - \log(2) \nonumber \\
        &\le - \inf_{\tilde{q} \in \tilde{F}} {\cal J}^\Theta_{{\bf w},1/2}(\tilde{q}). 
    \end{align}
    This implies that
    \begin{align} \label{eq:closed_set_size_bound}
        \limsup_{N \to \infty} \frac{1}{N^2} \log(|F_N|) \le \log(2) - \inf_{\tilde{q} \in \tilde{F}} {\cal J}^\Theta_{{\bf w},1/2}(\tilde{q}).
    \end{align}
    Similarly, for open sets $\tilde{U} \subseteq \tilde{{\cal W}}_\Theta$ we have
    \begin{align} \label{eq:open_set_size_bound}
        \liminf_{N \to \infty} \frac{1}{N^2} \log(|U_N|) &\ge \log(2) - \inf_{\tilde{q} \in \tilde{U}} {\cal J}^\Theta_{{\bf w},1/2}(\tilde{q}).
    \end{align}

    Fix $\epsilon > 0$. Since ${\cal U}$ is bounded, there is a finite set $R$ such that the intervals $\{(a,a+\epsilon): a \in R\}$ cover the range of ${\cal U}$. For each $a \in R$, let $\tilde{F}^a \coloneqq {\cal U}^{-1}([a,a+\epsilon])$. By continuity of ${\cal U}$, the $\tilde{F}^a$ are closed. Now, note that
    \begin{align}
        e^{N^2 \psi_N} = \sum_{G \in {\cal G}_N} e^{N^2 {\cal U}(\tilde{Q}(G))} \le \sum_{a \in R} e^{N^2(a+\epsilon)} |F^a_N|.
    \end{align}
    The last inequality comes from counting the directed graphs that induce a colored graphon orbit in $\tilde{F}^a_N$, which yields $|F^a|$, and bounding ${\cal U}$ in this interval. We can further bound this by
    \begin{align}
        e^{N^2 \psi_N} \le |R| \sup_{a \in R} e^{N^2(a + \epsilon)} |F^a_N|.
    \end{align}
    Using Eq. \eqref{eq:closed_set_size_bound}, this yields
    \begin{align}
        \limsup_{N \to \infty} \psi_N \le \sup_{a \in R} \left( a + \epsilon + \log(2) - \inf_{\tilde{q} \in \tilde{F}^a} {\cal J}^\Theta_{{\bf w},1/2}(\tilde{q}) \right).
    \end{align}
    Since each $\tilde{q} \in \tilde{F}^a$ satisfies ${\cal U}(\tilde{q}) \ge a$, we have
    \begin{align}
        \sup_{\tilde{q} \in \tilde{F}^a} ({\cal U}(\tilde{q}) + \log(2) - {\cal J}^\Theta_{{\bf w},1/2}(\tilde{q})) \ge \sup_{\tilde{q} \in \tilde{F}^a} (a + \log(2) - {\cal J}^\Theta_{{\bf w},1/2}(\tilde{q})) = a + \log(2) - \inf_{\tilde{q} \in \tilde{F}^a} {\cal J}^\Theta_{{\bf w},1/2}(\tilde{q}).
    \end{align}
    Substituting above,
    \begin{align}
        \limsup_{N \to \infty} \psi_N &\le \epsilon + \sup_{a \in R} \left( \sup_{\tilde{q} \in \tilde{F}^a} ({\cal U}(\tilde{q}) + \log(2) - {\cal J}^\Theta_{{\bf w},1/2}(\tilde{q})) \right) \nonumber \\
        &= \epsilon + \sup_{\tilde{q} \in \tilde{{\cal W}}_\Theta} ({\cal U}(\tilde{q}) + \log(2) - {\cal J}^\Theta_{{\bf w},1/2}(\tilde{q})).
    \end{align}
    Our rate function selects colored graphon orbits that have the correct coloring fraction, so we can write
    \begin{align}
        \limsup_{N \to \infty} \psi_N \le \epsilon + \sup_{\substack{\tilde{q} \in \tilde{{\cal W}}_\Theta \\ \tilde{c}[\tilde{q}] = \tilde{c}_{{\bf w}}}} ({\cal U}(\tilde{q}) + \log(2) - \tilde{{\cal I}}_{1/2}(\tilde{h}[\tilde{q}])) = \epsilon + \sup_{\substack{\tilde{q} \in \tilde{{\cal W}}_\Theta \\ \tilde{c}[\tilde{q}] = \tilde{c}_{{\bf w}}}} ({\cal U}(\tilde{q}) + \tilde{{\cal H}}(\tilde{h}[\tilde{q}])) .
    \end{align}
    where ${\cal H}$ is as defined in Equation \eqref{eq:graphon_entropy}. 

    Similarly, for each $a$, let $\tilde{U}^a \coloneqq {\cal U}^{-1}((a,a+\epsilon))$, which is open by continuity of ${\cal U}$. Note that
    \begin{align}
        e^{N^2 \psi_N} \ge \sup_{a \in R} e^{N^2 a} |U^a_N|.
    \end{align}
    By Eq. \eqref{eq:open_set_size_bound},
    \begin{align}
        \liminf_{N \to \infty} \psi_N \ge \sup_{a \in R} \left( a + \log(2) - \inf_{\tilde{q} \in \tilde{U}^a} {\cal J}^\Theta_{{\bf w},1/2}(\tilde{q}) \right).
    \end{align}
    Each $\tilde{q} \in \tilde{U}^a$ satisfies ${\cal U}(\tilde{q}) \le a + \epsilon$. Therefore, 
    \begin{align}
        \sup_{\tilde{q} \in \tilde{U}^a} ({\cal U}(\tilde{q}) + \log(2) - {\cal J}^\Theta_{{\bf w},1/2}(\tilde{q})) &\le \sup_{\tilde{q} \in \tilde{U}^a} (a + \epsilon + \log(2) - {\cal J}^\Theta_{{\bf w},1/2}(\tilde{q})) \nonumber \\
        &= a + \epsilon + \log(2) - \inf_{\tilde{q} \in \tilde{U}^a} {\cal J}^\Theta_{{\bf w},1/2}(\tilde{q}).
    \end{align}
    Substituting in our previous bound yields
    \begin{align}
        \liminf_{N \to \infty} \psi_N &\ge -\epsilon + \sup_{a \in R} \left( \sup_{\tilde{q} \in \tilde{U}^a} ({\cal U}(\tilde{q}) + \log(2) - {\cal J}^\Theta_{{\bf w},1/2}(\tilde{q})) \right) \nonumber \\
        &= -\epsilon + \sup_{\tilde{q} \in \tilde{{\cal W}}_\Theta} ({\cal U}(\tilde{q}) + \log(2) - {\cal J}^\Theta_{{\bf w},1/2}(\tilde{q})) \nonumber \\
        &= -\epsilon + \sup_{\substack{\tilde{q} \in \tilde{{\cal W}}_\Theta \\ \tilde{c}[\tilde{q}] = \tilde{c}_{{\bf w}}}} ({\cal U}(\tilde{q}) + \tilde{{\cal H}}(\tilde{h}[\tilde{q}])).
    \end{align}
    Since $\epsilon$ is arbitrary, we have $\psi_N \to \psi$, which completes the proof.
\end{proof}

\begin{theorem} \label{thm:colored_graphon_convergence}
     Let $\Tilde{M}^*$ be the set of maximizers of the variational problem \eqref{eq:colored_graphon_partition}. Let $Q_N$ be a colored graph on $N$ vertices drawn from the model implied by function ${\cal U}$. Then, for any $\eta > 0$, there exist $C, \kappa > 0$ such that, for any $N$,
    \begin{align*}
        \mathbb{P}\{ \delta^\Theta_\square(\Tilde{Q}_N, \Tilde{M}^*) > \eta \} \le C e^{-N^2 \kappa},
    \end{align*}
    where $\mathbb{P}$ denotes the probability measure implied by the model.
\end{theorem}

\begin{proof}
    This proof follows the same procedure as the proof of Theorem 3.2 in \citet{chatterjee_estimating_2013}.
\end{proof}

\subsection{Technical results}
This section contains results that are used in the proofs of the results in the previous section.

\begin{lemma} \label{lem:inverse_contraction_ldp}
    Let ${\cal X}$ and ${\cal Y}$ be compact metric spaces and $\Psi:{\cal X} \to {\cal Y}$ be a continuous surjection. Consider a sequence of probability measures $(\mu_N)_{N \in \mathbb{N}}$ on ${\cal X}$, and their pushfowards by $\Psi$: $\nu_N \coloneqq \mu_N \circ \Psi^{-1}$. Suppose that for all $N$ and $y \in \textrm{supp}(\nu_N)$, there exists a finite set $\Upsilon_N(y) \subset \Psi^{-1}(\{y\})$ such that:
    \begin{enumerate}
        \item (Finite support) There exist weights $\alpha_N(x,y) > 0$, $x \in \Upsilon_N(y)$, with $\sum_{x \in \Upsilon_N(y)} \alpha_N(x,y) = 1$ such that the $\Psi$-disintegration of $\mu_N$ is
        \begin{align}
            \mu_N(dx) = \int \mu_N^y(dx) \nu_N(dy), \quad \mu_N^y = \sum_{x \in \Upsilon_N(y)} \alpha_N(x,y) \delta_{x} \textrm{ for } y \in \textrm{supp}(\nu_N).
        \end{align}
        \item (Subexponential weight floor) The weights satisfy 
        \begin{align}
            \liminf_{N \to \infty} \inf_{y \in {\cal Y}} \min_{x \in \Upsilon_N(y)} \frac{1}{a_N} \log(\alpha_N(x,y)) = 0,
        \end{align}
        \item (Uniform fiber approximation) There exists a sequence $(\epsilon_N)_{N \in \mathbb{N}}$, with $\epsilon_N \to 0$, such that $\Upsilon_N(y)$ is an $\epsilon_N$ net of $\Psi^{-1}(y)$ for all $y \in \textrm{supp}(\nu_N)$.
    \end{enumerate}
    Then, if the sequence $(\nu_N)_{N \in \mathbb{N}}$ satisfies a large deviation principle on ${\cal Y}$ at speed $a_N$ with rate function $I_Y$, the sequence $(\mu_N)_{N \in \mathbb{N}}$ satisfies a large deviation principle on ${\cal X}$ at speed $a_N$ with rate function $I_X \coloneqq I_Y \circ \Psi$.
\end{lemma}

\begin{proof}
    The strategy for this proof is to show that the measures $\mu_N$ satisfy the condition for Bryc's Theorem (Theorem 4.4.2 in \citet{dembo_large_2010}). For a topological space ${\cal Z}$, let $C_b({\cal Z})$ denote the set of real-valued continuous bounded functions. We need to show that the sequence $\mu_N$ is exponentially tight, and that for all $\phi \in C_b({\cal X})$, the following holds:
    \begin{align} \label{eq:bryc_condition}
        \lim_{N \to \infty} \frac{1}{a_N} \log\left( \int e^{-a_N \phi} \, d\mu_N \right) = - \inf_{x \in {\cal X}} \{ \phi(x) + I_X(x) \}.
    \end{align}
    Note that the sequences $\mu_N$ and $\nu_N$ are automatically exponentially tight since the spaces are compact metric (this also automatically implies that all rate functions are good). Therefore, we just need to show that the limit above exists, and that it corresponds to the expression on the right. If this is the case, Bryc's theorem states that $\mu_N$ satisfies the LDP with speed $a_N$ and rate function $I_X$. 

    First, fix a function $\phi \in C_b({\cal X})$. Define the fiberwise infimum of $\phi$:
    \begin{align}
        \Phi(y) \coloneqq \inf\{ \phi(x) : \Psi(x)=y \}.
    \end{align}
    Now, note that since $\Psi$ is continuous surjective and ${\cal X}$ is compact, the correspondence $\Gamma(y) \coloneqq \Psi^{-1}(y)$ is non-empty, compact-valued and has a closed graph. Therefore, since $\phi$ is continuous, Berge's maximum theorem states that $\Phi \in C_b({\cal Y})$.

    By definition, we have that $\phi(x) \ge \Phi(\Psi(x))$ for all $x \in {\cal X}$. Therefore,
    \begin{align}
        \int e^{-a_N \phi(x)} \, d\mu_N(x) \le \int e^{-a_N \Phi(\Psi(x))} \, d\mu_N(x) = \int e^{-a_N \Phi(y)} \, d\nu_N(y).
    \end{align}
    Taking limits and applying Varadhan's Lemma (Theorem 4.3.1 in \citet{dembo_large_2010}) to the integral over ${\cal Y}$ yields
    \begin{align}
        \limsup_{N \to \infty} \frac{1}{a_N} \log\left( \int e^{-a_N \phi(x)} \, d\mu_N(x) \right) &\le \limsup_{N \to \infty} \frac{1}{a_N} \log\left( \int e^{-a_N \Phi(y)} \, d\nu_N(y) \right) \nonumber \\
        &= - \inf_{y \in {\cal Y}} \{ \Phi(y) + I_Y(y) \} \nonumber \\
        & = - \inf_{x \in {\cal X}} \{ \Phi(\Psi(x)) + I_Y(\Psi(x)) \} \nonumber \\
        &= - \inf_{x \in {\cal X}} \{ \phi(x) + I_X(x) \},
    \end{align}
    where the last equality used the infimum in the definition of $\Phi$. This yields the limsup bound.

    To obtain the liminf bound, we need the structure of the $\Psi$-disintegration of $\mu_N$. Fix some $\eta > 0$. Then, for all $y$ there exists some $x_0 \in \Psi^{-1}(y)$ such that  $\phi(x_0) \le \Phi(y) + \eta/2$. Since $\phi$ is continuous and ${\cal X}$ is compact, by the Heine-Cantor theorem, $\phi$ must be uniformly continuous. Together with the $\epsilon_N$ net condition, this implies that there is some $N_\eta$ (independent of $y$) such that for all $N > N_\eta$ there is some $x_1^N \in \Upsilon_N(y)$ that satisfies $\phi(x_1^N) \le \phi(x_0) + \eta/2$, which implies $\phi(x_1^N) \le \Phi(y) + \eta$. Now, note that
    \begin{align}
        \sum_{x \in \Upsilon_N(y)} \alpha_N(x,y) e^{-a_N \phi(x)} \ge \alpha_N(x_1^N,y) e^{-a_N \phi(x_1^N)} \ge \left[ \min_{x \in \Upsilon_N(y)}\alpha_N(x,y) \right] e^{-a_N [\Phi(y) + \eta]}.
    \end{align}
    Therefore, for $N$ large enough we obtain
    \begin{align}
        \int e^{-a_N \phi(x)} \, d\mu_N(x) &= \int \left( \sum_{x \in \Upsilon_N(y)} \alpha_N(x,y) e^{-a_N \phi(x)} \right) \, d\nu_N(y) \nonumber \\
        & \ge \int \left[ \min_{x \in \Upsilon_N(y)}\alpha_N(x,y) \right] e^{-a_N [\Phi(y) + \eta]} \, d\nu_N(y) \nonumber \\
        & \ge \left[ \inf_{y' \in {\cal Y}} \min_{x \in \Upsilon_N(y')}\alpha_N(x,y') \right] \int e^{-a_N [\Phi(y) + \eta]} \, d\nu_N(y).
    \end{align}
    Taking logs and limits, and applying Varadhan's Lemma to $\Phi + \eta$ yields
    \begin{align}
        &\liminf_{N \to \infty} \frac{1}{a_N} \log\left( \int e^{-a_N \phi(x)} \, d\mu_N(x) \right) \nonumber \\
        &\ge \liminf_{N \to \infty} \inf_{y' \in {\cal Y}} \min_{x \in \Upsilon_N(y)}\frac{1}{a_N}\log(\alpha_N(x,y')) - \eta - \inf_{y \in {\cal Y}} \{ \Phi(y) + I_Y(y) \}
    \end{align}
    Using the subexponential weight floor condition and noting that this bound holds for any $\eta$, we conclude that
    \begin{align}
        \liminf_{N \to \infty} \frac{1}{a_N} \log\left( \int e^{-a_N \phi(x)} \, d\mu_N(x) \right) &\ge - \inf_{y \in {\cal Y}} \{ \Phi(y) + I_Y(y) \} \nonumber \\
        &= - \inf_{x \in {\cal X}} \{ \phi(x) + I_X(x) \}.
    \end{align}

    Since the two bounds coincide, we have that the limit in Eq. \eqref{eq:bryc_condition} exists for all $\phi \in C_b({\cal X})$, and the expression is satisfied with $I_X = I_Y \circ \Psi$. Therefore, the conditions in Bryc's Theorem hold. This completes the proof.
\end{proof}

\begin{lemma} [Theorem 2.6 in \citet{varadhan_large_2010}] \label{lem:prod_measures_ldp}
    Let ${\cal X}$ and ${\cal Y}$ be Polish spaces. Let the sequences of measures $(\mu_N)_{N \in \mathbb{N}}$ on ${\cal X}$ and $(\nu_N)_{N \in \mathbb{N}}$ on ${\cal Y}$ satisfy large deviations principles at the same speed $a_N$ with good rate functions $I_X$ and $I_Y$, respectively. Then the sequence $(\mu_N \otimes \nu_N)_{N \in \mathbb{N}}$ satisfies a large deviations principle at speed $a_N$ with rate function
    \begin{align}
        I_{X \times Y}(x,y) = I_X(x) + I_Y(y).
    \end{align}
\end{lemma}

\begin{definition}
    For a directed graph $G=(V,A)$ and for $X,Y \subseteq V$, let $a_G(X,Y)$ denote the number of arcs that originate from a node in $X$ and terminate in a node in $Y$. Let
    \begin{align}
        d_G(X,Y) \coloneqq \frac{a_G(X,Y)}{|X||Y|}
    \end{align}
    denote the density of arcs from $X$ to $Y$.
\end{definition}

\begin{definition}
    Consider a directed graph $G = (V,A)$ and a partition ${\cal P} = \{V_1, \ldots, V_k\}$ of $V$. define the weighted directed graph $G_{{\cal P}}$ on $V$ as the complete directed graph whose arc $(u,v)$ has weight $d_G(V_i,V_j)$ if $u \in V_i$ and $v \in V_j$.
\end{definition}

\begin{lemma} \label{lem:weak_regularity}
    There exists a constant $C_0$ such that for every $\eta \in (0,1)$ and every directed graph $G=(V,A)$, $V$ has a partition ${\cal P}$ into $k \le \exp(C_0/\eta^2)$ classes such that
    \begin{align}
        d_\square(G, G_{\cal P}) \le \eta.
    \end{align}
\end{lemma}

\begin{proof}
    Fix $\eta \in (0,1)$. Let $N = |V|$ and let $B$ be the adjacency matrix of $G$ for some ordering of $[N] \coloneqq (1,\ldots,N)$. Consider the bipartite \textit{undirected} graph with adjacency matrix
    \begin{align}
        \tilde{B} =
        \begin{pmatrix}
            0 & B^\top \\
            B & 0
        \end{pmatrix}.
    \end{align}
    Define the left and right sets $L \coloneqq \{1,\ldots,N\}$ and $R \coloneqq \{N+1,\ldots,2N\}$. A consequence of the Weak Regularity Lemma for undirected graphs (Lemma 9.3 in \citet{lovasz_large_2012}), there is a partition ${\cal Q}$ of $[2N]$ into $k \le \exp(C_1/\eta^2)$ classes such that
    \begin{align}
        d_\square(\tilde{B},\tilde{B}_{\cal Q}) \le \frac{\eta}{4},
    \end{align}
    where $C_1$ is independent of the graph and of $N$. Since refining a partition can only decrease the cut distance to the original graph, we can consider the refinement ${\cal Q} \vee \{L,R\}$. This gives partitions ${\cal Q}_1$ and ${\cal Q}_2$ of $L$ and $R$ with $k_1$ and $k_2$ classes, respectively, with $k_1,k_2 \le k$. Note that ${\cal Q}_2$ induces a partition ${\cal Q}_2' = (Q_{2,i}')_{1\le i \le k_2}$ of $L$ given by
    \begin{align}
        Q_{2,i}' = \{ 2N+1 - x : x \in Q_{2,i} \}.
    \end{align}
    
    We can further consider the refinement ${\cal P} \coloneqq {\cal Q}_1 \vee {\cal Q}_2'$ of $L$, which satisfies $|{\cal P}| \le k_1 k_2$. We can extend ${\cal P}$ to a partition ${\cal Q}'$ of $L \cup R$. Let $P(x)$ and $Q'(x)$ be the partition classes of ${\cal P}$ and ${\cal Q}'$ that contain element $x$. Then we can define ${\cal Q}'$ by
    \begin{align}
        Q'(x) =
        \begin{cases}
            P(x) & \textrm{if } x \in L, \\
            P(2N+1-x) & \textrm{if } x \in R.
        \end{cases}
    \end{align}
    Note that ${\cal Q}'$ is a refinement of ${\cal Q}$. Therefore, we have that
    \begin{align}
        d_\square(\tilde{B},\tilde{B}_{{\cal Q}'}) \le d_\square(\tilde{B},\tilde{B}_{\cal Q}) \le \frac{\eta}{4}.
    \end{align}

    Returning to the directed graph $G$, the cut distance between $G$ and $G_{{\cal P}}$ can be written as
    \begin{align}
        d_\square(G,G_{{\cal P}}) = \frac{1}{N^2} \max_{S,T \subseteq L} \left| \sum_{i \in S, j \in T} (B_{ij}-(B_{\cal P})_{ij}) \right|.
    \end{align}
    Note that any cut $S,T$ can be implemented in the bipartite graph by choosing appropriate $S' \subseteq L$ and $T' \subseteq R$ (a cross-cut). Therefore,
    \begin{align}
        d_\square(G,G_{{\cal P}}) \le 4 \frac{1}{(2N)^2} \max_{S,T \subseteq L \cup R} \left| \sum_{i \in S, j \in T} (\tilde{B}_{ij}-(\tilde{B}_{{\cal Q}'})_{ij} \right| = 4 d_\square(\tilde{B},\tilde{B}_{{\cal Q}'}).
    \end{align}
    Therefore, for the partition ${\cal P}$, we have
    \begin{align}
        d_\square(G,G_{{\cal P}}) \le \eta.
    \end{align}
    Additionally, we have that
    \begin{align}
        |{\cal P}| \le k_1 k_2 \le k^2 \le \exp(2C_1/\eta^2).
    \end{align}
    This completes the proof.
\end{proof}

\begin{definition}
    For integer $N > 0$, for $i \in \{0, \ldots, N-1\}$ denote with $I^N_i$ the set $(i/N,(i+1)/N]$. An \textit{$N$-permutation} $\sigma_\tau$ is a measurable bijection determined by a permutation $\tau \in S_N$ such that
    \begin{align}
        \sigma_\tau(x) = x + \frac{\tau(i) - i}{N}
    \end{align}
    for $x \in I^N_i$. Let $\Sigma_N$ denote the set of $N$-permutations.
\end{definition}

\begin{lemma} \label{lem:fiber_density}
    There exist constants $\kappa_1, \kappa_2, N_0>0$ such that for all integer $N>N_0$, all colored graphons $Q$ over $N$ nodes, and all $\tilde{q} \in \tilde{{\cal W}}_\Theta$ such that $(\tilde{h}[\tilde{q}], \tilde{c}[\tilde{q}]) = (\tilde{h}^Q, \tilde{c}^Q)$, there exists an $N$-permutation $\hat{\sigma} \in \Sigma_N$ such that $q^{\hat{\sigma}} \coloneqq (\hat{\sigma} h^Q, c^Q)$ satisfies
    \begin{align}
        \delta_\square^\Theta(\tilde{q}^{\hat{\sigma}},\tilde{q}) \le \frac{\kappa_1}{\sqrt{\log(N)}} + 2 |\Theta| N^{-\kappa_2}.
    \end{align}
\end{lemma}

\begin{proof}
    To simplify notation, denote $h_1 \coloneqq h^Q$ and $c_1 \coloneqq c^Q$. Without loss of generality, we can set the coloring $c_1$ such that the sets $c_1^{-1}(\theta)$ are pairwise disjoint intervals. This is because the condition $(\tilde{h}[\tilde{q}], \tilde{c}[\tilde{q}]) = (\tilde{h}^Q, \tilde{c}^Q)$ and the distance $\delta_\square^\Theta(\tilde{q}^{\sigma_0},\tilde{q})$ only involve orbits. Denote with ${\cal K} = \{K_{\theta_1}, \ldots, K_{\theta_{|\Theta|}}\}$ the partition of $[0,1]$ into the intervals corresponding to this coloring.
    
    Since the orbits of the graphon and the coloring match, there is a representative $q_2 = (h_2, c_2) \in \tilde{q}$ such that $c_2 = c_1$. Additionally, there exists a measure-preserving bijection $\sigma_{12}$ such that $h_2 = \sigma_{12} h_1$. Now, note that for all $\sigma_0 \in \Sigma_N$,
    \begin{align}
        \delta_\square^\Theta(\tilde{q}^{\sigma_0}, \tilde{q}) = \inf_{\sigma} \left\{ d_\square(\sigma_0 h_1, \sigma h_2) + d_\Theta(c_1, \sigma c_2) \right\},
    \end{align}
    where the infimum is over measure-preserving bijections. I will bound this distance by searching over color-preserving bijections. That is, $\sigma$ such that $\sigma(K_\theta) = K_\theta$ for all $\theta$. We have that
    \begin{align}
        \delta_\square^\Theta(\tilde{q}^{\sigma_0}, \tilde{q}) \le \inf_{\sigma \textrm{ color-preserving}} \left\{ d_\square(\sigma_0 h_1, \sigma h_2) + d_\Theta(c_1, \sigma c_2) \right\} = \inf_{\sigma \textrm{ color-preserving}} d_\square(\sigma_0 h_1, \sigma h_2).
    \end{align}

    Fix $\eta > 0$. By Lemma \ref{lem:weak_regularity}, there exists a partition ${\cal P}$ of $\{1,\ldots,N\}$ with $k \le \exp(C_0/\eta^2)$ classes such that, for the directed graph $G(Q)$,
    \begin{align}
        d_\square(G(Q),G(Q)_{\cal P}) \le \eta.
    \end{align}
    For the graphon $h_1$, this implies that there is a partition ${\cal J}$ of $[0,1]$, where each partition class is composed of unions of intervals $I^N_i$, such that the graphon $u$ given by
    \begin{align}
        u(x,y) = \frac{1}{|J_a||J_b|} \int_{J_a \times J_b} h_1(x',y') \, dx' \, dy' \quad \textrm{for } (x,y) \in J_a \times J_b
    \end{align}
    satisfies
    \begin{align}
        d_\square(h_1,u) \le \eta.
    \end{align}
    Since the cut distance is invariant under measure-preserving bijections,  we have that for any $\sigma_0, \sigma$,
    \begin{align}
        d_\square(\sigma_0 h_1, \sigma_0 u), d_\square(\sigma \circ \sigma_{12} h_1, \sigma \circ \sigma_{12} u) \le \eta.
    \end{align}
    Therefore, by the triangle inequality, for any fixed color-preserving measure-preserving bijection $\sigma$
    \begin{align}
        \delta_\square^\Theta(\tilde{q}^{\sigma_0}, \tilde{q}) \le d_\square(\sigma_0 h_1, \sigma h_2) \le 2 \eta + d_\square(\sigma_0 u, \sigma \circ \sigma_{12} u).
    \end{align}

    To control the distance $d_\square(\sigma_0 u, \sigma \circ \sigma_{12} u)$, let us first choose a color-preserving measure-preserving bijection $\sigma$ such that the family $(\sigma \circ \sigma_{12}(J_a) \cap K_\theta)_{1 \le a \le k}$ are pairwise disjoint intervals (modulo null sets) within $K_\theta$. Let $\sigma_{\textrm{int}} \coloneqq \sigma \circ \sigma_{12}$. Now I construct an $N$-permutation that approximates the transport properties of $\sigma_{\textrm{int}}$. Define the transport matrix
    \begin{align}
        x_{a\theta} = \lambda(J_a \cap \sigma_{\textrm{int}}^{-1}K_{\theta})),
    \end{align}
    where $\lambda$ is the Lebesgue measure. Intuitively, this captures how much mass from $J_a$ is sent to $K_{\theta}$ under $\sigma_{\textrm{int}}$. Additionally, define $p_{a} \coloneqq \lambda(J_a)$ and $w_\theta \coloneqq \lambda(K_\theta)$. Note that
    \begin{align}
        \sum_{\theta} x_{a\theta} = p_{a}, \quad \sum_{a} x_{a\theta} = w_\theta.
    \end{align}
    Now, since $N p_{a}$ and $N w_\theta$ are integers for all $a, \theta$ (because they originate from a graphon over $N$ vertices), by Lemma 3 in \citet{baranyai_factorization_1974} there exists a matrix $(n_{a\theta})_{a,\theta}$ such that
    \begin{align}
        n_{a\theta} \in \{\lfloor N x_{a\theta} \rfloor, \lceil N x_{a\theta} \rceil\}, \quad \sum_{\theta} n_{a\theta} = N p_{a}, \quad \sum_{a} n_{a\theta} = N w_\theta.
    \end{align}
    Define $y_{a\theta} = n_{a\theta}/N$. Note that $|y_{a\theta} - x_{a\theta}| \le 1/N$. Now define an $N$-permutation $\hat{\sigma}$ that sends $n_{a\theta}$ intervals $I^N_i$ from $J_a$ to $K_{\theta}$. This permutation exists by the feasibility result above. Under this construction,
    \begin{align}
        \lambda(J_a \cap \hat{\sigma}^{-1}(K_{\theta})) = y_{a\theta}.
    \end{align}
    Furthermore, the sets $\hat{\sigma}(J_a) \cap K_\theta$ can be chosen to be intervals, since only their measure is constrained. Furthermore, these can follow the same order as the intervals $\sigma_{\textrm{int}}(J_a) \cap K_\theta$.

    Now consider a single interval $K_\theta$. Within this interval, the intervals $\hat{\sigma}(J_a) \cap K_\theta$ and $\sigma_{\textrm{int}}(J_a) \cap K_\theta$ follow the same order. Without loss of generality, assume they are ordered according to $a$. Define the set
    \begin{align}
        E_\theta \coloneqq \{x \in K_\theta: \hat{\sigma}^{-1}(x) \textrm{ belongs to a different class than } \sigma_{\textrm{int}}^{-1}(x)\}.
    \end{align}
    Since the measure of $E$ is determined by the discrepancies in the boundaries of the intervals, we have that
    \begin{align}
        \lambda(E_\theta) &\le \sum_{a=1}^k \Bigg\{\underbrace{\left| \sum_{b=1}^{a-1} (x_{b\theta} - y_{b\theta}) \right|}_{\textrm{left endpoints}} + \underbrace{\left| \sum_{b=1}^{a} (x_{b\theta} - y_{b\theta}) \right|}_{\textrm{right endpoints}} \Bigg\} \nonumber \\
        &\le \sum_{a=1}^k \left\{ \sum_{b=1}^{a-1} |x_{b\theta} - y_{b\theta}| + \sum_{b=1}^{a} |x_{b\theta} - y_{b\theta}| \right\} \nonumber \\
        &\le \frac{1}{N} \sum_{a=1}^k (2a-1) \nonumber \\
        &= \frac{k^2}{N}
    \end{align}
    Letting $E \coloneqq \bigcup_{\theta \in \Theta} E_\theta$, we get that the total discrepancy in the classes of the preimages is
    \begin{align}
        \lambda(E) = \sum_{\theta} \lambda(E_\theta) \le \frac{|\Theta| k^2}{N}.
    \end{align}

    Since $\sigma_{\textrm{int}} u$ and $\hat{\sigma} u$ have the same value on the sets $\sigma_{\textrm{int}}^{-1}(J_a) \times \sigma_{\textrm{int}}^{-1}(J_b)$ and $\hat{\sigma}^{-1}(J_a) \times \hat{\sigma}^{-1}(J_b)$, we have that they must match on $E^c \times E^c$. Therefore, since $|u| \le 1$, we obtain
    \begin{align}
        \lVert \sigma_{\textrm{int}} u - \hat{\sigma} u \rVert_1 &= \int_{[0,1]^2 \backslash E^c \times E^c} |\sigma_{\textrm{int}} u - \hat{\sigma} u| \le 1-(1-\lambda(E))^2 \le 2 \lambda(E) \le \frac{2|\Theta| k^2}{N}.
    \end{align}
    Since the $L^1$ norm upper-bounds the cut norm, we have that
    \begin{align}
        d_\square(\sigma_{\textrm{int}} u, \hat{\sigma} u) \le \frac{2|\Theta| k^2}{N}.
    \end{align}

    Going back to our distance of interest, we have that
    \begin{align}
        \delta_\square^\Theta(\tilde{q}^{\hat{\sigma}}, \tilde{q}) \le 2 \eta + d_\square(\sigma_{\textrm{int}} u, \hat{\sigma} u) \le 2 \eta + \frac{2|\Theta| k^2}{N}.
    \end{align}
    Choose $\eta = C_1/\sqrt{\log(N)}$ for some $C_1 > 0$. This yields
    \begin{align}
        \delta_\square^\Theta(\tilde{q}^{\hat{\sigma}}, \tilde{q}) \le \frac{2C_1}{\sqrt{\log(N)}} + \frac{2|\Theta| N^{2C_0/C_1^2}}{N}.
    \end{align}
    Choosing $C_1 > \sqrt{2 C_0}$, we obtain the desired bound.
\end{proof}

\end{document}